\documentclass[12 pt, onecolumn, draftclsnofoot]{IEEEtran}
\usepackage{enumerate}
\usepackage{graphicx}
\usepackage{epsf}
\usepackage{subfigure}
\usepackage{amsmath}
\usepackage{amssymb}
\usepackage{array}
\usepackage{setspace}
\usepackage[amsmath,thmmarks,amsthm]{ntheorem}
\usepackage[ruled,lined]{algorithm2e}
\usepackage{algorithmic}
\usepackage{color}
\usepackage{amsfonts}
\usepackage{dsfont}
\usepackage{xfrac}
\usepackage{breqn}
\usepackage{bbm}
\usepackage{cite}
\usepackage{multirow}
\usepackage{float}
\usepackage{diagbox}
\usepackage{mathtools}
\usepackage{adjustbox}

\graphicspath{{fig/}}

\theoremseparator{.}

\newtheorem{lemma}{Lemma}

\newtheorem{prop}{Proposition}

\begin{document}

\title{\huge CoMP Enhanced Subcarrier and Power Allocation for Multi-Numerology based 5G-NR Networks}



\author{\IEEEauthorblockN{Li-Hsiang Shen, Chia-Yu Su, and Kai-Ten Feng}\\
\IEEEauthorblockA{Department of Electrical and Computer Engineering \\
National Yang Ming Chiao Tung University, Hsinchu, Taiwan\\
gp3xu4vu6.cm04g@nctu.edu.tw, su0760214.cm07g@nctu.edu.tw and ktfeng@mail.nctu.edu.tw}}

\maketitle

\begin{abstract}
With proliferation of fifth generation (5G) new radio (NR) technology, it is expected to meet the requirement of diverse traffic demands. We have designed a coordinated multi-point (CoMP) enhanced flexible multi-numerology (MN) for 5G-NR networks to improve the network performance in terms of throughput and latency. We have proposed a CoMP enhanced joint subcarrier and power allocation (CESP) scheme which aims at maximizing sum rate under the considerations of transmit power limitation and guaranteed quality-of-service (QoS) including throughput and latency restrictions. By employing difference of two concave functions (D.C.) approximation and abstract Lagrangian duality method, we theoretically transform the original non-convex nonlinear problem into a solvable maximization problem. Moreover, the convergence of our proposed CESP algorithm with D.C. approximation is analytically derived with proofs, and is further validated via numerical results. Simulation results demonstrated that our proposed CESP algorithm outperforms the conventional non-CoMP and single numerology mechanisms along with other existing benchmarks in terms of lower latency and higher throughput under the scenarios of uniform and edge users.
\end{abstract}


\section{Introduction}

With proliferation of abundant and diverse services in the future wireless networks, the fifth generation new radio (5G-NR) technology is expected to meet the requirement of three service types including enhanced mobile broadband (eMBB), massive machine-type communications (mMTC) and ultra-reliable and low-latency communications (URLLC) \cite{1_5G}, \cite{2_B5G}. They require tens-of-Gbps peak transmission rates, capability of connections of million devices and ultra-low latency at milli-second level, respectively \cite{3_5G}, which are undoubtedlessly challenging for existing communication systems due to its inflexibility and stringent service architecture. To realize dynamic and flexible adaptation to various service requirements, 5G-NR standardized by the 3rd generation partnership project (3GPP) organization has introduced an advanced multi-numerology (MN) under orthogonal frequency division multiplexing (OFDM) techniques, which enables co-existence of different frame structures of numerologies based on the service types \cite{4_3GPP_38211}.

Numerology refers to a set of parameters of subcarrier spacing (ScS) for frequency domain and symbol duration for time domain in an OFDM system. However, to facilitate multiple services, the corresponding configurations of the frame structure should be readjusted. Therefore, multi-numerology-enabled network (MNN) of the envisioned 5G-NR techniques has attracted substantial attentions \cite{5_flex, 6_RAPS, 7_ROFN, 8_latency_MN} thanks to its potential achievement of accommodation for heterogeneous services by properly arranging numerology based on service requirements. Generally, the set of numerology with longer symbol duration can be configured for delay-tolerant transmissions, whilst shorter symbol duration is more feasible for urgent service such as URLLC in order to meet rigorous latency requirement. To control overhead in MNN, the authors of \cite{5_flex} have proposed a heuristic method to obtain efficient configuration of numerologies under a feasible number of framing length. Under a fixed multi-numerology structure, paper \cite{6_RAPS} has designed a packet scheduling algorithm with the considerations of fairness among user equipments (UEs). However, it potentially provokes inflexibility to readjust transmission frame structure to simultaneously fulfill multiple service requirements. Nevertheless, in paper \cite{7_ROFN}, they have designed an enhanced scheduling scheme for flexible numerologies in both frequency and time domains where the original non-solvable problem is decoupled and approximated as a convex problem. In \cite{8_latency_MN}, the authors have proposed a maximum-minimum Knapsack-based resource allocation problem for multi-user transmissions constrained by latency requirement under a mixed numerology. In our previous work of \cite{previouswork}, we have designed to optimize power allocation for rate maximization under a fixed configured MNN, which reflects the substantial influence of inter-numerology interference (INI). However, the co-existence of multi-numerology with different ScS will eradicate orthogonality property imposing severe INI \cite{9_INI_2, 10_main_INI, previouswork}, which is not considered in the open literatures of \cite{5_flex, 6_RAPS, 7_ROFN, 8_latency_MN}. 


The coordinated multi-point (CoMP) transmission is a promising technique to improve transmission rate performance under a 5G-NR based MNN interfered by INI. By employing CoMP transmission, users will be associated with multiple base stations (BSs) and are capable of mitigating the interferences by coordinated transmissions among BSs. In \cite{13_CoMP_JT}, the authors have considered an user-centric CoMP-based joint transmission (JT) clustering problem by obtaining an optimal parameter of power level difference which determines a set of BSs participating CoMP transmission. The paper \cite{14_CoMP_Availability} has conducted the availability analysis for a single service type based on the CoMP technique by formulating an optimization problem for finding the optimal resource allocation policy. In \cite{15_CoMP_GCoMP}, the authors have derived the system throughput outage probability with various clustering mechanisms for a single-numerology CoMP-enabled non-orthogonal multiple access network. In \cite{16_DPA_CoMP} and \cite{17_RA_CoMP}, they consider a power allocation problem under single-numerology CoMP transmissions by designing energy efficient BS on-off schemes. In paper \cite{18_main_ref}, the authors have conceived and evaluated an advanced power-domain non-orthogonal transmission and dual connectivity scheme for a fixed-framing CoMP-enabled network. The papers of \cite{13_CoMP_JT, 14_CoMP_Availability, 15_CoMP_GCoMP, 16_DPA_CoMP, 17_RA_CoMP, 18_main_ref} consider the fixed transmission frame structure of single-numerology CoMP, which leads to inflexibility and unavailability of other services. However, to support diverse multi-services, CoMP-enabled multi-numerology should be considered by coordinating radio resources among transmit BSs. Therefore, in order to improve system spectrum efficiency, resource allocation of power allocation (PA) and subcarrier assignment (SA) is a substantial problem for CoMP-enhanced multi-numerology transmissions. In \cite{21_RA_OSA}, the authors only consider SA problem by proposing a heuristic optimization method considering inter-cell interference. The authors in \cite{19_RA_J_SAPA} have solved a joint SA and PA problem considering inter-cell interference. Moreover, the authors in \cite{20_RA_J_O_BS, 18_main_ref, 22_Alg_SA} have designed an optimization problem aiming to maximize the energy efficiency by jointly solving the user association in addition to SA and PA with the consideration of inter-cell interference. However, the INI occurrence in multiple services for a CoMP-enabled MNN is not considered in papers of \cite{21_RA_OSA, 19_RA_J_SAPA, 20_RA_J_O_BS, 18_main_ref, 22_Alg_SA}, which potentially deteriorates system performance for supporting multi-services.


\begin{table}
	\centering
	\small
	\caption {Comparison Table of Literatures}
	\begin{tabular}{l|cccccccccc}
		\hline
		\small		
		Literatures &\cite{5_flex, 6_RAPS, 7_ROFN, 8_latency_MN} & \cite{previouswork} & \cite{13_CoMP_JT,17_RA_CoMP} & \cite{15_CoMP_GCoMP,16_DPA_CoMP} & \cite{14_CoMP_Availability},\cite{18_main_ref} & \cite{21_RA_OSA} & \cite{19_RA_J_SAPA, 20_RA_J_O_BS, 22_Alg_SA} & Proposed \\ \hline \hline
		\small
		Multi-Numerology & \checkmark & \checkmark &  &  & 	& 	&   & \checkmark \\
		CoMP & 	& & \checkmark	& \checkmark	& \checkmark	& 	&   & \checkmark \\		
		Power Allocation &	& \checkmark &	& \checkmark& \checkmark	& & \checkmark	& \checkmark \\
		Subcarrier Assignment & &	& & & \checkmark & \checkmark & \checkmark & \checkmark \\
		User Association  &	& & & & \checkmark & \checkmark & \checkmark & \checkmark \\ 
		Throughput   &	& \checkmark & & & \checkmark & \checkmark & \checkmark & \checkmark \\		
		Latency &	& & & & \checkmark & \checkmark & \checkmark & \checkmark \\
		Quality-of-Service & &	& & & \checkmark & \checkmark & \checkmark & \checkmark \\
		
		\hline
	\end{tabular} \label{RefComparison}
\end{table}


A comparison of the related work is summarized in Table \ref{RefComparison}. However, we notice that none of the aforementioned works jointly consider resource allocation for CoMP enhanced MN system. As mentioned in the previous paragraph, with the flexible ScS assignment in multi-numerology enabled 5G NR networks, we can assign larger ScS with shorter duration to satisfy latency-aware service, which provides higher flexibility of both frequency and time resources. Additionally, CoMP is capable of mitigating the inter-cell interference (ICI) \cite{23_CoMP_urllc} by transforming interference into desired signals improving system throughput. Moreover, higher system throughput and lower latency can be achieved by applying CoMP in MNN, which is not considered in recent studies and our previous work in \cite{previouswork}. In others words, CoMP enhanced MNN can potentially supports simultaneous eMBB and URLLC applications, which strikes a compelling tradeoff between throughput and latency requirements. Therefore, inspired by the aforementioned literatures, we have conceived the first architecture of CoMP enhanced MN for 5G NR networks. We jointly consider INI and ICI respectively coming from MN and CoMP transmissions and design subcarrier and power allocation problem for a CoMP-enhanced MNN. The main contributions of this paper are summarized in the followings.
\begin{itemize}
	\item We have firstly conceived a CoMP enhanced multi-numerology 5G NR network which support simultaneous eMBB and URLLC services. With flexible ScS structure of MN, larger ScS with shorter time can be selected to meet the stringent latency requirement and vice versa. Additionally, throughput performance can be improved by employing CoMP techniques which mitigate ICI from multiple BSs. We design a subcarrier and power allocation problem which aims at maximizing system sum rate guaranteed by maximum allowable transmit power and user's quality-of-service (QoS) constraints including throughput and latency.
	
	\item We proposed a CoMP enhanced subcarrier and power allocation (CESP) algorithm for a 5G NR MNN. In CESP algorithm, we transform the original problem into two sub-problems, i.e., subcarrier assignment and power allocation. Both sub-problems employ the difference of two concave function (D.C.) approximation converting the non-convex and nonlinear problem into a solvable convex one. Additionally, to deal with binary variables of SA, we adopt abstract Lagrangian duality method by introducing a new constraint as a penalty term to the objective function. Moreover, the convergence of each sub-problem employing D.C. approximation is proven. It is also theoretically proved that the proposed CESP algorithm can iteratively improve  system performance until convergence.
	
	\item We have evaluated system performance of our proposed CESP for CoMP enhanced MNN. It can be observed that CoMP with MN achieves higher sum data rate and lower delay-throughput outage compared to the conventional non-CoMP single numerology scenario for uniform and edge users under different service requirements. Benefited from CoMP enhanced MNN, the proposed CESP algorithm can achieve higher throughput performance guaranteed by QoS and latency, which effectively supports simultaneous eMBB and URLLC services.
\end{itemize}

The remainder of this paper is organized as follows. In Section \ref{SYS_MOD}, we introduce our CoMP enhanced MN system model consisting of multi-BSs and multi-users. We also conceive a subcarrier and power allocation problem aiming at maximizing system sum rate. The proposed CESP algorithm is illustrated in Section \ref{Alg} including sub-schemes of PA and SA, respectively. Furthermore, the convergence analysis of proposed CESP algorithm is demonstrated in Section \ref{CON_ANA}. Simulation results are provided in Section \ref{PER_EVA}, whilst conclusions are drawn in Section \ref{SEC_CON}.

\section{System Model and Problem Formulation} \label{SYS_MOD}

\subsection{System Model} 

\begin{figure}
	\centering
	\includegraphics[width=4in]{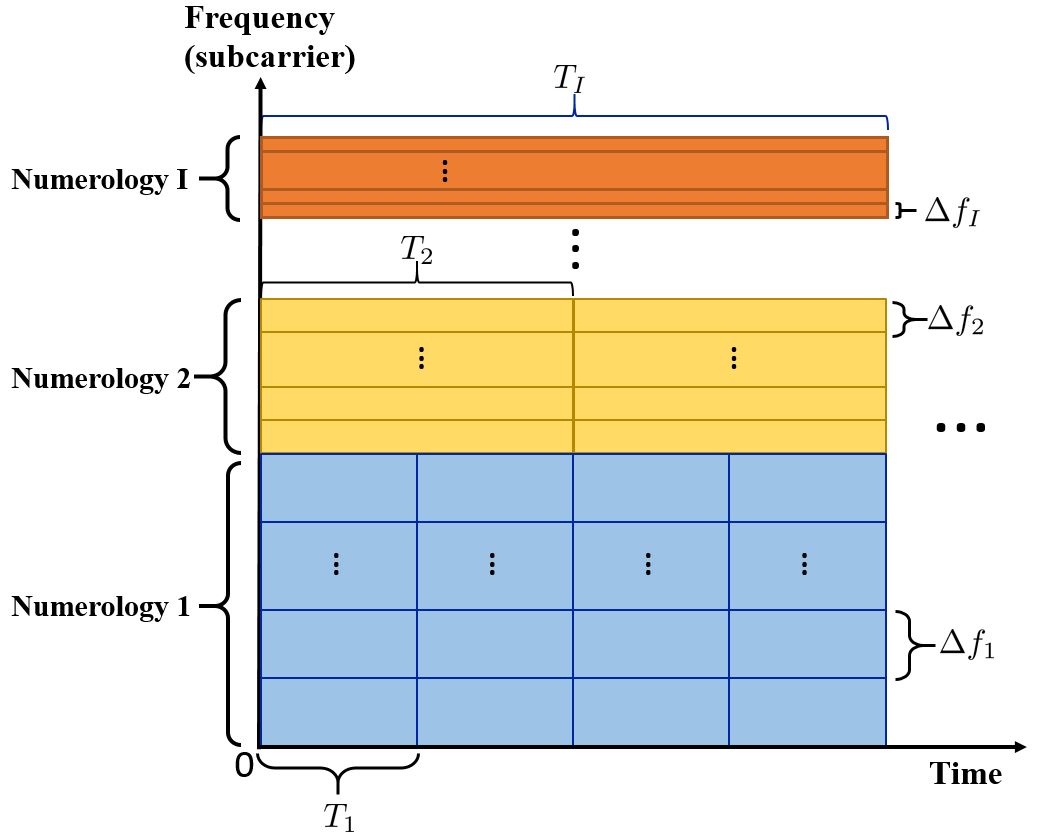}
	\caption{The illustration of MN frame structure.}
	\label{Fig.MN}
\end{figure}

We consider a 5G NR multi-numerology system consisting of a set of BSs $\mathcal{K}=\{1, ..., K\}$ and users $\mathcal{M}=\{1, ..., M\}$ served by numerology set of $\mathcal{I}=\{1, ..., I\}$. The illustration of MN frame structure is depicted in Fig. \ref{Fig.MN}, where it contains subcarrier set of $\mathcal{N}=\{1, ..., N\}$ with each ScS denoted as $\Delta f_{i}$ and symbol duration defined as $T_{i}$. It is noted that the multiplication of ScS and symbol duration is a constant implying that a larger ScS of $\Delta f_{i}$ is accompanied with a shorter symbol duration $T_{i}$. Based on 5G NR specification of \cite{4_3GPP_38211}, numerology $i$ has a subcarrier spacing of $\Delta f_{i} = 2^{\mu_{i}} \cdot 15 \text{ kHz}$, where $\mu_{i} \in \{0,1,...,5\}$. In this paper, we consider two numerologies in our system which is the basic case in a MN system, i.e., $\mathcal{I}=\{1,2\}$. Note that the proposed two-numerology model can be readily extend to a general MN architecture by selecting different $\mu_{i}$. Moreover, for analysis simplicity, we assume that $\Delta f_{2}>\Delta f_{1}$. We define the ratio between two numerologies is $Q=\frac{\Delta f_{2}}{\Delta f_{1}}$. We consider $N_i$ as the number of available subcarriers for numerology $i$ and total subcarrier of difference numerologies is $\sum_{i \in \mathcal{I}} N_{i} = N$. The architecture of proposed CoMP enhanced multi-numerology downlink network is illustrated in Fig. \ref{Fig.SM}. We consider $\boldsymbol{P} =\{ p_{k,m,n}^{i} \}$ as the transmit power set, where $p_{k,m,n}^{i}$ denotes the serving power of BS $k$ to user $m$ on numerology $i$'s subcarrier $n$. We defined $h_{k,m,n}^{i}$ as the channel condition measured from BS $k$ to user $m$ on numerology $i$'s subcarrier $n$, where $h_{k,m,n}^{i}=g_{k,m} \beta_{k,m,n}^{i}$. Note that $g_{k,m}$ represents the distance-based large scale fading, whereas $\beta_{k,m,n}^{i}$ indicates small scale fading effect. Furthermore, 
we denote a resource assignment set as $\boldsymbol{X} = \{ x_{k,m,n}^{i} \in \{0, 1\} \}$ indicating subcarrier assignment, i.e., $x_{k,m,n}^{i}=1$ if numerology $i$'s subcarrier $n$ of BS $k$ is assigned to user $m$ and vice versa.

\begin{figure}
	\centering
	\includegraphics[width=4.5in]{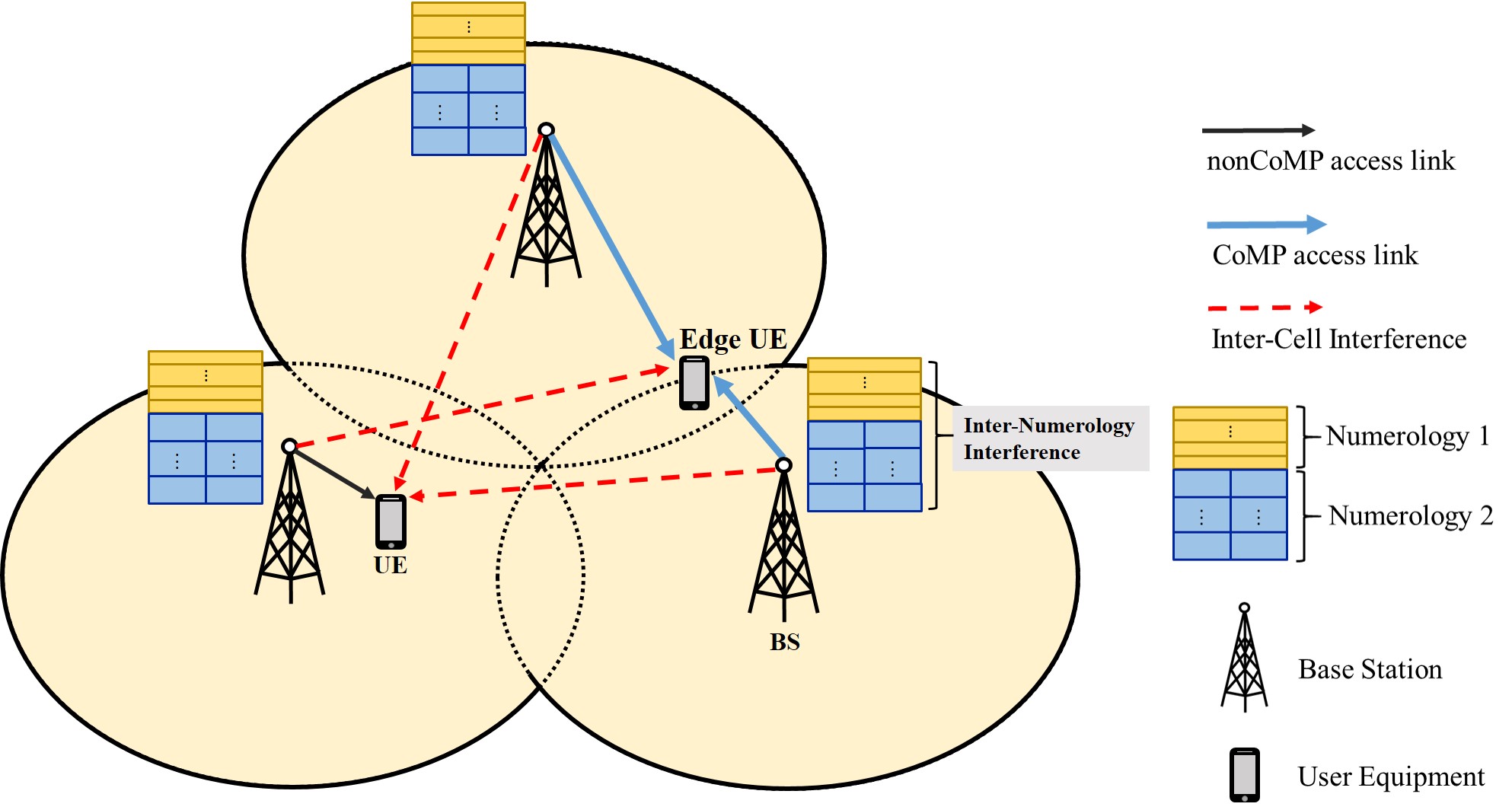}
	\caption{The architecture of proposed CoMP enhanced multi-numerology downlink network.}
	\label{Fig.SM}
\end{figure}

\subsubsection{Signal Model}
By employing CoMP techniques, joint transmission is performed such that the user is capable of receiving superposed signals from more than one BS. Therefore, the signal-to-interference-plus noise ratio (SINR) of user $m$ served by BS $k$ on numerology $i$'s subcarrier $n$ can be formulated as
\begin{equation}\label{eq:SINR}
\Gamma_{k,m,n}^{i} = \frac{\sum_{k' \in \mathcal{K}} p_{k',m,n}^{i} x_{k',m,n}^{i} |h_{k',m,n}^{i}|^2}{\sum_{j\in \mathcal{K} \setminus k}\left( x_{j,m,n}^{i}|h_{j,m,n}^{i}|^2 \cdot INI_{k,m,n}^{i} \right) + ICI_{k,m,n}^{i} + (\sigma_{k,m,n}^{i})^2 },
\end{equation}
where $INI_{k,m,n}^{i}$ represents the inter-numerology interference induced by non-orthogonality of MN architecture, which is given by \cite{10_main_INI}
\begin{equation}\label{eq:INI}
INI_{k,m,n}^{i}=
\begin{cases}
	\sum_{v=1}^{N_2}\frac{p^{2}_{k,m,v}}{N_1^2} \left| \frac{\sin\left[\frac{\pi}{Q}\left(Q\left(v-1\right)-\left(n-1\right)\right)\right]}{\sin\left[\frac{\pi}{2N_1}\left(Q\left(v-1\right)-\left(n-1\right)+N_1\right)\right]} \right|^2, \text{if } i=1, n=1,...,N_1, \\
	\sum_{v=1}^{N_1}\frac{p^{1}_{k,m,v}}{2 N_1^2} \left| \frac{\sin\left[\frac{\pi}{Q}\left(\left(v-1\right)-Q\left(n-1\right)\right)\right]}{\sin\left[\frac{\pi}{2N_1}\left(\left(v-1\right)-Q\left(n-1\right)-N_1\right)\right]} \right|^2, \text{if } i=2, n=1,...,N_2.
\end{cases}
\end{equation} 
Based on the consideration of $\Delta f_2>\Delta f_1$, we can further obtain 
$INI^{i}_{k,m,n}=INI^{i}_{k,m,\hat{n}}$, where $\hat{n} =\left( n \bmod  N_2\right)$ for $i=2, \forall n=(q-1)N_2+1, ..., qN_2, \forall q=2, ..., Q$. Moreover, $ICI^{i}_{k,m,n}$ in $\eqref{eq:INI}$ indicates the inter-cell interference coming from the other BSs using the same numerology, which is expressed as
\begin{equation}\label{eq:ICI}
ICI_{k,m,n}^{i} = \sum_{k'\in \mathcal{K} \setminus k}\sum_{m'\in \mathcal{M} \setminus m} p_{k',m',n}^{i} x_{k',m',n}^{i} |h_{k',m',n}^{i}|^2.
\end{equation}
We further notice that INI comes from different numerologies, i.e., the subcarrier assigned to the user $m$ with different numerologies will induce INI, whilst ICI is led by the same subcarrier allocated to different users. According to the SINR model in $\eqref{eq:SINR}$, we can acquire individual data rate of user $m$ which is served by BS $k$ on numerology $i$'s subcarrier $n$ as
\begin{equation}\label{eq:r}
r^{i}_{k,m,n} = \log_2\left(1+\Gamma^{i}_{k,m,n}\right) ,
\end{equation} 
and the total achievable throughput is then given by
\begin{equation}\label{eq:SR}
\Upsilon = \sum_{i\in \mathcal{I}} \sum_{k\in \mathcal{K}}\sum_{m\in \mathcal{M}}\sum_{n\in\mathcal{N}} r^{i}_{k,m,n} .
\end{equation} 

\subsubsection{CoMP User Association}

We adopt channel conditions between serving BS and users to determine whether users are served by CoMP transmissions. We firstly sort the obtained channel state information $\left|h^{i}_{k,m,n}\right|^{2}$. We define $\kappa_{m}(j)$ as the $j$-th strongest channel index for user $m$, therefore the sorting channel for user $m$ can be written as $\left| h^{i}_{\kappa_{m}(j),m,n} \right|^{2}$, which implies that $\left| h^{i}_{\kappa_{m}(1),m,n} \right|^{2} \geq \left| h^{i}_{\kappa_{m}(2),m,n} \right|^{2} \geq ... \geq \left| h^{i}_{\kappa_{m}(K),m,n} \right|^{2}$. We define a channel level difference threshold of $\sigma_{m}$ to determine the BSs serving user $m$ \cite{13_CoMP_JT}, which is given by
\begin{equation} \label{eq:comp}
x^{i}_{\kappa_{m}(j),m,n} = 
\begin{cases}
1, \text{if } x^{i}_{\kappa_{m}(1),m,n} = 1, \frac{\left| h^{i}_{\kappa_{m}(1),m,n} \right|^{2}}{\left| h^{i}_{\kappa_{m}(k),m,n} \right|^{2}} \leq \sigma_{m}, \forall k \in \mathcal{K}\setminus \kappa_{m}(1), \\
0 , \text{otherwise}.
\end{cases}
\end{equation}
The CoMP condition in $\eqref{eq:comp}$ indicates that we transform the potential strong interference to serving BSs in order to improve signal quality. Moreover, we only decide CoMP transmission links between BSs and users, whilst the resource allocation indicator of numerology types and subcarrier is not determined, which should be designed considering different service constraints in the following subsection.

\subsubsection{Latency Model}

In an MNN as explained previously, it is envisioned that the numerology with larger ScS providing shorter symbol duration is more appropriate for services with stringent latency \cite{8_latency_MN}. Therefore, we define the average latency of user $m$ as 
\begin{equation}\label{latency_a}
l_{m}=\frac{\sum_{i\in \mathcal{I}} \sum_{k\in \mathcal{K}}\sum_{n\in\mathcal{N}} \left( x_{k,m,n}^{i} \cdot l^{i}\right)}{\sum_{i\in \mathcal{I}} \sum_{k\in \mathcal{K}}\sum_{n\in\mathcal{N}} x_{k,m,n}^{i}},
\end{equation}
where $l^i$ indicates latency of numerology $i$. Accordingly, users with stringent latency requirement will be preferably served by assigning the numerology with shorter timeslot, and vice versa. However, for a CoMP enhanced MNN, it becomes a compellingly challenging to obtain optimum assignment for multi-BSs serving multi-users, which should be designed considering different service constraints in the following subsection.

\subsection{Problem Formulation}
Based on the designed signal model for the CoMP enhanced MNN, we aim at maximizing the total system sum rate by assigning power and subcarriers considering interferences of ICI and INI. Moreover, the optimization problem is constrained by allowable maximum transmit power and user QoS including throughput and latency, which can be formulated as 
\begin{subequations} \label{P0}
	\begin{align}
	\max_{\boldsymbol{X}, \boldsymbol{P}} \ & \Upsilon \label{eq:obj_func}\\ 
	\text{s.t.} \ & x^{i}_{k,m,n}\in \{0,1\}, \qquad \forall i\in \mathcal{I},k\in \mathcal{K},m\in \mathcal{M},n\in \mathcal{N}, \label{C1}\\
	& \sum_{\forall m \in \mathcal{M}} x^{i}_{k,m,n} = \{0, 1\}, \quad \forall i\in \mathcal{I},k\in \mathcal{K},n\in \mathcal{N}, \label{C2}\\ 
	& 0\leq p^{i}_{k,m,n}\leq p^{max}, \quad \forall i\in \mathcal{I},k\in \mathcal{K},m\in \mathcal{M},n\in \mathcal{N}, \label{C3}\\ 
	& \sum_{\forall i \in \mathcal{I}} \sum_{\forall m \in \mathcal{M}}\sum_{\forall n \in \mathcal{N}} x^{i}_{k,m,n} p^{i}_{k,m,n} \leq P^{max}_{k}, \quad \forall k\in \mathcal{K}, \label{C4}\\ 
	& \sum_{\forall i \in \mathcal{I}} \sum_{\forall k \in \mathcal{K}} \sum_{\forall n \in \mathcal{N}} r_{k,m,n}^{i} \geq r_{m}^{req}, \quad \forall m\in \mathcal{M}, \label{C5}\\ 
	& l_{m} \leq l_{m}^{req}, \quad \forall m\in \mathcal{M}. \label{C6} 
	\end{align}
\end{subequations}
Note that the arguments of $\boldsymbol X$ and $\boldsymbol P$ in $\eqref{eq:obj_func}$ indicate subcarrier assignment and power allocation, respectively. Constraint $\eqref{C1}$ indicates resource assignment indicator should be a binary variable, and constraint $\eqref{C2}$ confines that each subcarrier is only assigned to one user. In $\eqref{C3}$, the transmit power of BS $k$ serving user $m$ on each numerology $i$'s subcarrier $n$ is constrained by $p^{max}$, and constraint $\eqref{C4}$ demonstrates the maximum allowable transmit power of BS $k$ which is denoted as $P_{k}^{max}$. Constraint $\eqref{C5}$ indicates that each user $m$ needs to satisfy the minimum required data rate of $r^{req}_{m}$, whereas constraint of $\eqref{C6}$ denotes the latency requirement of each user $m$. We can observe that the optimization problem $\eqref{P0}$ is non-convex and nonlinear due to the resource assignment indicator functions and power of signal, ICI and INI terms in $\eqref{eq:SINR}$, which leads to the difficulties to directly obtain the theoretical optimal solution. Therefore, in the following section, we transform the original problem into a theoretically solvable one and propose a CoMP enhanced subcarrier and power allocation scheme to achieve optimum outcomes.

\section{Proposed CoMP Enhanced Subcarrier and Power Allocation (CESP) Algorithm} \label{Alg}

The original optimization problem of $\eqref{P0}$ provokes a difficulty to acquire the optimum due to its non-convexity and nonlinearity. Therefore, we propose a CESP algorithm which iteratively resolves power allocation and subcarrier assignment in separate sub-problems until convergence is achieved \cite{24_iter_alg}. We can then obtain the candidate optimum solution through iterations which is described as 
\begin{equation} \label{eq:iter}
\cdots \rightarrow \overbrace{\underbrace{\mathcal{P}^{t}(\boldsymbol X^{t-1}, \boldsymbol P)}_{\text{SA of CESP}} \rightarrow \underbrace{\mathcal{P}^{t}(\boldsymbol X, \boldsymbol P^{t})}_{\text{PA of CESP}}}^{\text{Iteration } t} \rightarrow \cdots
\end{equation}
In $\eqref{eq:iter}$, $\mathcal{P}^{t}(\boldsymbol X, \boldsymbol P)$ is our objective problem at the $t$-th iteration including variables of subcarrier and power to be optimized. During iteration $t$, we can obtain the candidate solution of power allocation in PA sub-problem with fixed subcarrier at previous iteration, which is expressed as $\mathcal{P}^{t}(\boldsymbol X^{t-1}, \boldsymbol P)$. On the other hand, subcarrier assignment is determined in the SA sub-problem of $\mathcal{P}^{t}(\boldsymbol X^{t}, \boldsymbol P)$ given fixed power allocation results. Therefore, based on our original problem of $\eqref{P0}$, the PA sub-problem $\mathcal{P}^{t}(\boldsymbol X, \boldsymbol P^{t-1})$ is given by
	\begin{subequations} \label{PA}
		\begin{align}
		\max_{\boldsymbol P} \ & \Upsilon \label{eq:obj_func_PA}\\ 
		\text{s.t.} \ & \eqref{C3}, \eqref{C4}, \eqref{C5},
		\end{align}
	\end{subequations}
whilst the SA sub-problem $\mathcal{P}^{t}(\boldsymbol X^{t}, \boldsymbol P)$ is expressed as
	\begin{subequations} \label{SA}
		\begin{align}
		\max_{\boldsymbol X} \ & \Upsilon \label{eq:obj_func_SA}\\ 
		\text{s.t.} \ & \eqref{C1}, \eqref{C2}, \eqref{C4}, \eqref{C5}, \eqref{C6}.
		\end{align}
	\end{subequations}
However, we can observe from both sub-problems that they still possess nonlinearity and non-convexity properties, which should be resolved in order to obtain the optimum resource assignments. The PA and SA algorithms of our proposed CESP are designed in the following subsections.

\subsection{Power Allocation for Proposed CESP Algorithm}\label{PA_sec}
We firstly consider PA sub-problem $\mathcal{P}^{t}(\boldsymbol X^{t-1}, \boldsymbol P)$ in $\eqref{PA}$ with fixed subcarrier outcome $\boldsymbol{X}^{t-1}$ in previous iteration. Due to the non-convexity of the objective function and constraint $\eqref{C5}$, we employ D.C. approximation method \cite{20_RA_J_O_BS} to transform the optimization problem $\eqref{PA}$ into a solvable concave one. Therefore, the objective function of PA can be equivalently rewritten in a D.C. form as  
\begin{equation} \label{equTG}
\Upsilon = T(\boldsymbol P) - G(\boldsymbol P),
\end{equation}
where $T(\boldsymbol P)$ and $G(\boldsymbol P)$ are two concave functions respectively given by 
\begin{equation}\label{T}
T(\boldsymbol P)= \sum_{i\in \mathcal{I}} \sum_{k\in \mathcal{K}}\sum_{m\in \mathcal{M}}\sum_{n\in\mathcal{N}} \log_2\left( {\psi}_{k,m,n}^{i,(1)}\right),
\end{equation}
and
\begin{equation}\label{G}
G(\boldsymbol P)=\sum_{i\in \mathcal{I}} \sum_{k\in \mathcal{K}}\sum_{m\in \mathcal{M}}\sum_{n\in\mathcal{N}} \log_2\left( {\psi}_{k,m,n}^{i,(2)}\right),
\end{equation}
with
\begin{align}
	{\psi}_{k,m,n}^{i,(1)} & = \sum_{k' \in \mathcal{K}} p^{i}_{k',m,n} x^{i}_{k',m,n}|h^{i}_{k',m,n}|^2 + \sum_{j \in \mathcal{K} \setminus k}\left( x^{i}_{j,m,n}|h^{i}_{j,m,n}|^2 \cdot INI^{i}_{k,m,n} \right) + ICI^{i}_{k,m,n} + (\sigma^{i}_{k,m,n})^2, \label{psi_1}\\
	 {\psi}_{k,m,n}^{i,(2)} &= \sum_{j\in \mathcal{K} \setminus k}\left( x^{i}_{j,m,n}|h^{i}_{j,m,n}|^2 \cdot INI^{i}_{k,m,n} \right) + ICI^{i}_{k,m,n} + (\sigma^{i}_{k,m,n})^2.\label{psi_2}
\end{align}
Since the subtraction of two logarithmic functions in $\eqref{equTG}$ is not definitely concave, we perform approximation for the second term of $G(\boldsymbol P)$ to a linear function related to previous updated power solutions by adopting the first-order Taylor approximation \cite{18_main_ref}, which is represented as
\begin{equation}
	\widehat{G}(\boldsymbol P) \approx G(\boldsymbol P^{t-1})+ \nabla^{\dagger} G(\boldsymbol P^{t-1})\cdot(\boldsymbol P-\boldsymbol P^{t-1}),
\end{equation}
where $\nabla^{\dagger} G(\boldsymbol P)$ is a transposed first-order derivative vector of $G(\boldsymbol P)$ with length $2K\cdot N$ and $\dagger$ is the transpose operation. Note that $\boldsymbol P^{t-1}$ is the power allocation outcome of the previous iteration. The corresponding matrix element of $\nabla G(\boldsymbol P)$ is given by
\begin{equation}
\begin{aligned}
\frac{\partial G(\boldsymbol P)}{\partial p_{k,m,n}^{i}} &= \sum_{i'\in \mathcal{I} \setminus i} \sum_{m'\in \mathcal{M}}{\sum_{n'\in \mathcal{N}}{\frac{C^{i'}_{n'} x_{k,m',n'}^{i'}|h_{k,m',n'}^{i'}|^2 }{{\psi}_{k,m',n'}^{i',(2)}}}} \\
	& \!+\! \sum_{i'\in \mathcal{I} \setminus i} \sum_{k' \in \mathcal{K} \setminus k}{\sum_{m' \in \mathcal{M} \setminus m}}{\left(\frac{\sum_{m'' \in \mathcal{M} \setminus m'} x_{k,m'',n}^{i} |h_{k,m',n}^{i}|^2}{{\psi}_{k',m',n}^{i,(2)}} \!+\! \sum_{n'\in \mathcal{N}}{\frac{C^{i'}_{n'} x_{k,m',n'}^{i'} |h_{k,m',n'}^{i'}|^2 }{{\psi}_{k,m',n'}^{i',(2)}}}\right)},
\end{aligned}
\end{equation}
where ICI related term is further derived as
\begin{equation}
C^{i'}_{n'} = 
\begin{cases}
	\frac{1}{N_1^2}\left|\frac{\sin\left[\frac{\pi}{Q}(Q( (n \bmod N_{2})-1)-(n'-1))\right]}{\sin\left[\frac{\pi}{2N_1}(Q( (n \bmod N_{2}) - 1) - (n'-1)+N_1)\right]}\right|^2, \text{if } i'=1, \\
	\frac{1}{2 N_1^2}\left|\frac{\sin\left[\frac{\pi}{Q}((n-1)-Q( (n' \bmod N_{2})-1))\right]}{\sin\left[\frac{\pi}{2N_1}((n-1)-Q( (n' \bmod N_{2}) - 1 ) -N_1) \right]}\right|^2, \text{if } i'=2.
\end{cases}
\end{equation}
Similarly, the minimum rate requirement constraint in $\eqref{C5}$ can be transformed to a concave function by exploiting D.C. approximation. By Taylor approximation, we can have $\widehat{G}_{2}(\boldsymbol P) \approx G_2(\boldsymbol P^{t-1})+ \nabla^{\dagger} G_2(\boldsymbol P^{t-1})\cdot(\boldsymbol P-\boldsymbol P^{t-1})
$, where the previous results of $G_2(\boldsymbol P^{t-1})$ is generally obtained from $G_2(\boldsymbol P) = \sum_{\forall i \in \mathcal{I}} \sum_{\forall k \in \mathcal{K}} \sum_{\forall n \in \mathcal{N}} \log_2\left({\psi}_{k,m,n}^{i,(2)}\right) $ and the gradient of $\nabla^{\dagger} G_2(\boldsymbol P)$ is derived as
\begin{equation}
\begin{aligned}
\frac{\partial G_2(\boldsymbol P)}{\partial p_{k,m,n}^{i}}=
\begin{cases}
	\sum_{i'\in \mathcal{I} \setminus i} \left( \sum_{n'\in \mathcal{N}} \frac{C^{i'}_{n'} x^{i'}_{k,m,n'}|h^{i'}_{k,m,n'}|^2 }{\psi^{i'}_{k,m,n'}}+
\frac{\sum_{m' \in \mathcal{M} \setminus m} x^{i}_{k,m',n}|h^{i}_{k,m,n}|^2}{{\psi}^{i,(2)}_{k,m,n}} 
\right), \\ \qquad \text{if } x^{i}_{k',m,n}=1, \forall k'\in \mathcal{K} \setminus k, \forall i'\neq i,
\\
	\sum_{i'\in \mathcal{I} \setminus i} \sum_{n'\in \mathcal{N}} \frac{C^{i'}_{n'} x^{i'}_{k,m,n'}|h^{i'}_{k,m,n'}|^2 }{\psi^{i',(2)}_{k,m,n'}}, \\ \qquad \text{if } x^{i}_{k',m,n}=1, \forall k'=k, \forall i'\neq i.
\end{cases}
\end{aligned}
\end{equation}
After applying D.C. approximation in PA sub-problem of $\eqref{eq:obj_func_PA}$ and QoS constraint in $\eqref{C5}$, we can transform the original problem $\eqref{PA}$ into a concave problem formulated as
\begin{subequations} \label{PA_DC}
	\begin{align}
	\max_{\boldsymbol P} \ & T(\boldsymbol P)-\widehat{G}(\boldsymbol P) \label{PA_obj_dc}\\ 
	\text{s.t.} \ & \eqref{C3}, \eqref{C4}, \\ 
	& \sum_{\forall i \in \mathcal{I}} \sum_{\forall k \in \mathcal{K}} \sum_{\forall n \in \mathcal{N}} \log_2\left( \psi^{i,(1)}_{k,m,n} \right) - \widehat{G}_{2}(\boldsymbol P) \geq r^{req}_{m}, \quad \forall m\in \mathcal{M}. \label{C5_dc}
	\end{align}
\end{subequations}
Accordingly, we can obtain the optimum solution of power allocation in PA sub-problem by using arbitrary convex optimization software tools.

\subsection{Subcarrier Assignment for Proposed CESP Algorithm} \label{SA_sec}

After obtaining the candidate power allocation, we proceed to determine the subcarrier assignment based on the SA sub-problem in $\eqref{SA}$. To solve the SA sub-problem, we firstly adopt integer relaxation method \cite{18_main_ref} to convert the discrete subcarrier assignment variable $x^{i}_{k,m,n}=\{0,1\}$ to a continuous parameter, i.e., $\hat{x}^{i}_{k,m,n} \in [0,1]$. Accordingly, the original constraints of $\eqref{C1}$ and $\eqref{C2}$ can be rewritten as 
\begin{equation}
0 \leq \hat{x}^{i}_{k,m,n} \leq 1 \label{C1-1},
\end{equation}
and
\begin{equation}
\sum_{\forall m \in \mathcal{M}} \hat{x}^{i}_{k,m,n} \leq 1, \label{C2-1}
\end{equation}
respectively. However, the solution of subcarrier assignment using constraints of $\eqref{C1-1}$ and $\eqref{C2-1}$ will lead to potential quantization error and cannot achieve original optimum. Accordingly, in order to acquire the optimum while keeping all parameters continuous, motivated by the method in \cite{22_Alg_SA}, we introduce an auxiliary constraint which is given by
\begin{equation}
	\sum_{i\in \mathcal{I}} \sum_{k\in \mathcal{K}}\sum_{m\in \mathcal{M}}\sum_{n\in\mathcal{N}} \left[ \hat{x}^{i}_{k,m,n} - \left(\hat{x}^{i}_{k,m,n}\right)^2 \right] \leq 0. \label{SA new_C} 
\end{equation}
Therefore, the original problem will be unchanged, which is stated in the following lemma.
\begin{lemma}
	With an auxiliary constraint of $\eqref{SA new_C}$, the original discrete variable of subcarrier assignment $x^{i}_{k,m,n}$ is equivalently converted to a continuous $0\leq \hat{x}^{i}_{k,m,n} \leq 1$ with unchanged solution of $\{0,1\}$.
\end{lemma}
\begin{proof}
We can observe from $\eqref{SA new_C}$ that $\hat{x}^{i}_{k,m,n}-\left(\hat{x}^{i}_{k,m,n}\right)^2$ performs the concavity property with the maximum point when $\hat{x}^{i}_{k,m,n}=\frac{1}{2}$ and two roots at $\hat{x}^{i}_{k,m,n}=\{0,1\}$. Therefore, $\eqref{SA new_C}$ guarantees $\hat{x}^{i}_{k,m,n}\leq 0$ and $\hat{x}^{i}_{k,m,n} \geq 1$. By combining the solution of $\eqref{SA new_C}$ in $\eqref{C1-1}$, i.e., $0\leq \hat{x}^{i}_{k,m,n}\leq 1$, we can equivalently have the discrete binary parameter of $x^{i}_{k,m,n}=\{0,1\}$ but sustains $x^{i}_{k,m,n}$ as a continuous variable as $\hat{x}^{i}_{k,m,n}$. Accordingly, the original problem becomes unchanged by using integer relaxation with proposed auxiliary constraint of $\eqref{SA new_C}$, which completes the proof.
\end{proof}
After introducing $\eqref{SA new_C}$ to the original SA sub-problem $\eqref{SA}$, the equivalent SA optimization problem can be further expressed as 
\begin{subequations} \label{P_SA_wNewC}
	\begin{align}
	\max_{\boldsymbol X} \ & \Upsilon \\ 
	\text{s.t.} \ & \eqref{C4}, \eqref{C5}, \eqref{C6}, \eqref{C1-1}, \eqref{C2-1}, \eqref{SA new_C}.
	\end{align}
\end{subequations}
Notice that the constraint $\eqref{SA new_C}$ is non-convex. Therefore, we employ the abstract Lagrangian duality \cite{26_Dual_Lagr} by transforming the constraint $\eqref{SA new_C}$ into a penalty term of original objective in SA sub-problem $\eqref{P_SA_wNewC}$. The abstract Lagrangian based SA objective can be obtained as
\begin{equation}
	\mathcal{L}(\boldsymbol{X}, \lambda) = \Upsilon - \lambda \sum_{i\in \mathcal{I}} \sum_{k\in \mathcal{K}}\sum_{m\in \mathcal{M}}\sum_{n\in\mathcal{N}} \left[ \hat{x}^{i}_{k,m,n} - \left(\hat{x}^{i}_{k,m,n}\right)^2 \right],
\end{equation}
where $\lambda\geq 0$ is a penalty parameter. Note that $\boldsymbol{X}$ now indicates the solution of continuous variable $\hat{x}^{i}_{k,m,n}$ instead of the discrete one $x^{i}_{k,m,n}$. Consequently, the transformed SA problem is reformulated as
\begin{subequations} \label{P_eq_SA}
	\begin{align}
		\max_{\boldsymbol{X}}& \min_{\lambda} \ \mathcal{L}(\boldsymbol{X}, \lambda) \\ 
		\text{s.t.} \ & \eqref{C4}, \eqref{C5}, \eqref{C6}, \eqref{C1-1}, \eqref{C2-1}. \label{SA-C-n}
	\end{align}
\end{subequations}
We define $\mathcal{D}$ as the feasible set comprising all constraints in $\eqref{SA-C-n}$. The problem in $\eqref{P_eq_SA}$ can be further denoted as $\max_{\boldsymbol{X}\in \mathcal{D}} \min_{\lambda} \mathcal{L}(\boldsymbol{X}, \lambda)$, and its corresponding dual problem is expressed as $\min_{\lambda} \max_{\boldsymbol{X}\in \mathcal{D}} \mathcal{L}(\boldsymbol{X}, \lambda)$. In the following, we theoretically prove that equivalence of SA optimization problem in $\eqref{P_SA_wNewC}$ and $\eqref{P_eq_SA}$.
\begin{prop}
	For a sufficiently large $\lambda$, the SA optimization problem $\eqref{P_SA_wNewC}$ is equivalent to $\eqref{P_eq_SA}$.
\end{prop}
	\begin{proof}
		The optimum of primal problem of $\eqref{P_SA_wNewC}$ is obtained as 
		\begin{equation} \label{Lp1}
		p^* = \max_{\boldsymbol X \in \mathcal{D}} \min_{\lambda} \ \mathcal{L}(\boldsymbol{X}, \lambda), 
		\end{equation}
		whereas the candidate solution of dual problem is acquired as
		\begin{equation} \label{Lp2}
		d^* = \min_{\lambda} \max_{\boldsymbol X \in \mathcal{D}} \ \mathcal{L}(\boldsymbol{X}, \lambda). 
		\end{equation}
		We define $\mu(\lambda) \triangleq \max_{\boldsymbol X \in \mathcal{D}} \mathcal{L}(\boldsymbol{X}, \lambda)$. According to the weak duality property and solutions in $\eqref{Lp1}$ and $\eqref{Lp2}$, we have the following equality as
		\begin{equation} \label{weak_duality}
		p^* = \max_{\boldsymbol X \in \mathcal{D}} \min_{\lambda} \mathcal{L}(\boldsymbol{X},\lambda) \leq \min_{\lambda}\mu(\lambda) = d^*.
		\end{equation} 
		Therefore, it comes up with two potential cases which are demonstrated as belows.
		\begin{itemize}
		\item \textbf{Case 1}:
		Suppose that we consider $\sum_{i\in \mathcal{I}} \sum_{k\in \mathcal{K}}\sum_{m\in \mathcal{M}}\sum_{n\in\mathcal{N}} \left[ \hat{x}^{i}_{k,m,n} - \left(\hat{x}^{i}_{k,m,n}\right)^2 \right]=0$ for obtaining the optimal solution. Let $d^*$ and $\lambda^*$ be the feasible solutions of $\eqref{P_SA_wNewC}$. Substituting $\lambda^*$ into the optimization problem of $\eqref{P_SA_wNewC}$ yields
		\begin{equation}
		d^* = \mu(\lambda^*) = \max_{\boldsymbol X \in \mathcal{D}} \min_{\lambda} \mathcal{L}(\boldsymbol{X},\lambda) = p^*.
		\end{equation} 
		Moreover, $\mu(\lambda)$ is a monotonically-decreasing function with respect of $\lambda$. Due to the equality of $d^* = \min_{\lambda} \mu(\lambda)$, we can therefore obtain
		\begin{equation}
		d^* = \mu(\lambda), \forall \lambda \geq \lambda^*,
		\end{equation}  
		which implies that the solution of (\ref{P_eq_SA}) results in the optimal solution of (\ref{P_SA_wNewC}) for sufficiently large values of $\lambda, \forall \lambda \geq \lambda^*$.
		
		\item \textbf{Case 2}: 
		Suppose that we consider $0 \leq \hat{x}^{i}_{k,m,n} \leq 1$, it reveals that $\sum_{i\in \mathcal{I}} \sum_{k\in \mathcal{K}}\sum_{m\in \mathcal{M}} \sum_{n\in\mathcal{N}} \\ \left[ \hat{x}^{i}_{k,m,n} - \left(\hat{x}^{i}_{k,m,n}\right)^2 \right] \geq 0$. Under inequality when the above function is greater than $0$, $\mu(\lambda^*)$ will have a tendency of approaching $-\infty$, which contradicts the derivation of weak duality equality in $\eqref{weak_duality}$. Accordingly, it should be $\sum_{i\in \mathcal{I}} \sum_{k\in \mathcal{K}}\sum_{m\in \mathcal{M}}\sum_{n\in\mathcal{N}} \left[ \hat{x}^{i}_{k,m,n} - \left(\hat{x}^{i}_{k,m,n}\right)^2 \right] = 0$ at the optimal point under $0 \leq \hat{x}^{i}_{k,m,n} \leq 1$.
		\end{itemize}
	Based on Cases 1 and 2, we can infer that  the SA optimization problem $\eqref{P_SA_wNewC}$ is equivalent to $\eqref{P_eq_SA}$ if sufficiently large $\lambda$ is given, which completes the proof.
	\end{proof}
However, we can know from problem $\eqref{P_eq_SA}$ that the objective function $\mathcal{L}(\boldsymbol{X},\lambda)$ and the subcarrier assignment constraint $\eqref{C5}$ are still not in concave forms. Hence, similar to PA sub-problem mentioned in previous subsection, we employ D.C. approximation to convert non-concave problem into a solvable concave one. First of all, we rewrite the objective function $\mathcal{L}(\boldsymbol{X}, \lambda)$ in a D.C. form which is expressed as
\begin{equation} \label{L_dc}
\mathcal{L}(\boldsymbol{X}, \lambda) = E_1(\boldsymbol{X}, \lambda) - E_2(\boldsymbol{X}, \lambda),
\end{equation} 
where $E_1(\boldsymbol{X}, \lambda)$ and $E_2(\boldsymbol{X}, \lambda)$ are respectively given by
\begin{equation}
E_1(\boldsymbol{X}, \lambda) = T(\boldsymbol X) - \lambda\left(\sum_{i\in \mathcal{I}} \sum_{k\in \mathcal{K}}\sum_{m\in \mathcal{M}}\sum_{n\in\mathcal{N}} \hat{x}_{k,m,n}^{i} \right), 
\end{equation}
and
\begin{equation}
E_2(\boldsymbol{X}, \lambda) = G(\boldsymbol X)- \lambda\left[\sum_{i\in \mathcal{I}} \sum_{k\in \mathcal{K}}\sum_{m\in \mathcal{M}}\sum_{n\in\mathcal{N}} \left(\hat{x}_{k,m,n}^{i}\right)^2 \right],
\end{equation}
with $T(\boldsymbol X)=\sum_{i\in \mathcal{I}} \sum_{k\in \mathcal{K}}\sum_{m\in \mathcal{M}}\sum_{n\in\mathcal{N}} \log_2\left( {\psi}_{k,m,n}^{i,(1)}\right)$ and $G(\boldsymbol X)=\sum_{i\in \mathcal{I}} \sum_{k\in \mathcal{K}}\sum_{m\in \mathcal{M}}\sum_{n\in\mathcal{N}} \\ \log_2\left( {\psi}_{k,m,n}^{i,(2)}\right)$. By substituting $x^{i}_{k,m,n}=\hat{x}^{i}_{k,m,n}$, the variables ${\psi}_{k,m,n}^{i,(1)}$ and ${\psi}_{k,m,n}^{i,(2)}$ are obtained in $\eqref{psi_1}$ and $\eqref{psi_2}$, respectively. Due to non-concavity property of $E_2(\boldsymbol{X}, \lambda)$, we exploit the first-order Taylor approximation, which is derived as
\begin{equation}
	\widehat{E}_2(\boldsymbol X, \lambda) \approx E_2(\boldsymbol X^{t-1}, \lambda)+ \nabla^{\dagger} E_2(\boldsymbol X^{t-1}, \lambda)\cdot(\boldsymbol X-\boldsymbol X^{t-1}),
\end{equation}
where $\nabla^{\dagger} E_{2}(\boldsymbol X, \lambda)$ is a transposed first-order derivative vector of $E_{2}(\boldsymbol X, \lambda)$ with a  length of $2K\cdot M \cdot N$. The corresponding matrix element of $\nabla E_{2}(\boldsymbol X, \lambda)$ can be obtained as
\begin{equation}
\begin{aligned}
\frac{\partial E_{2}(\boldsymbol X, \lambda)}{\partial \hat{x}^{i}_{k,m,n}}
=\frac{INI^{i}_{k,m,n}|h^{i}_{k,m,n}|^2}{\psi^{i,(2)}_{k,m,n}} + \sum_{k' \in \mathcal{K} \setminus k} \sum_{m' \in \mathcal{M} \setminus m} \left( \frac{p^{i}_{k,m',n}|h^{i}_{k,m',n}|^2}{\psi^{i,(2)}_{k,m',n}}\right) - 2\lambda \hat{x}^{i}_{k,m,n}.
\end{aligned}
\end{equation}
Similarly, the non-concave constraint $\eqref{C5}$ is also transformed to a concave function by adopting D.C. approximation as
\begin{equation}
\widehat{G}_{2}(\boldsymbol X) \approx G_2(\boldsymbol X^{t-1})+ \nabla^{\dagger} G_2(\boldsymbol X^{t-1})\cdot(\boldsymbol X-\boldsymbol X^{t-1}),
\end{equation}
where $G_2(\boldsymbol X) = \sum_{\forall i \in \mathcal{I}} \sum_{\forall k \in \mathcal{K}} \sum_{\forall n \in \mathcal{N}} \log_2\left({\psi}_{k,m,n}^{i,(2)}\right) $ and the gradient of $\nabla^{\dagger} G_2(\boldsymbol X)$ is
\begin{equation}
\frac{\partial G_2(\boldsymbol X)}{\partial \hat{x}^{i}_{k,m,n}}
=\frac{INI^{i}_{k,m,n}|h^{i}_{k,m,n}|^2}{\psi^{i, (2)}_{k,m,n}}. 
\end{equation} 
Therefore, the SA sub-problem is transformed into a concave optimization problem represented by
\begin{subequations} \label{P_SA_final}
	\begin{align}
	\max_{\boldsymbol X} \ & E_1(\boldsymbol{X}, \lambda) - \widehat{E}_2(\boldsymbol{X},\lambda) \\ 
	\text{s.t.} \ & \eqref{C3}, \eqref{C4}, \eqref{C6}, \eqref{C1-1}, \eqref{C2-1}, \\
	& \sum_{\forall i \in \mathcal{I}} \sum_{\forall k \in \mathcal{K}} \sum_{\forall n \in \mathcal{N}} \log_2\left( \psi^{i,(1)}_{k,m,n} \right) - \widehat{G}_{2}(\boldsymbol X) \geq r^{req}_{m}, \quad \forall m\in \mathcal{M}. \label{SA_C5_dc}
	\end{align}
\end{subequations}
We can obtain the optimum solution of power allocation in PA sub-problem by using arbitrary convex optimization software tools. The concrete procedure of proposed CESP algorithm is demonstrated in Algorithm \ref{alg}.
We firstly initialize the solution set of $\{\boldsymbol{X}, \boldsymbol{P}\}$, iteration counter index $t=1$, and maximum allowable iteration number $T_{max}$. With employment of D.C. approximation, the optimum solutions can be obtained in each sub-problem. The iterative algorithms CESP-SA and CESP-PA are conducted for power allocation in $\eqref{PA_DC}$ and subcarrier assignment in $\eqref{P_SA_final}$, respectively until convergence. Note that the convergence condition holds when the difference of system sum rate between two iterations is smaller than a given threshold $\gamma_{thr}$. During the process of CESP-SA at the $t$-th iteration, we can acquire the optimal subcarrier assignment $\boldsymbol{X}^{t}$ with fixed candidate power allocation in previous iteration $\boldsymbol{P}^{t-1}$, and vice versa. Accordingly, we can obtain the optimum solution set of subcarrier assignment and power allocation as $\{\boldsymbol{X}^*, \boldsymbol{P}^* \}= \{\boldsymbol{X}^t,\boldsymbol{P}^t\}$.

\begin{algorithm}[!tb]
	\caption{Proposed CESP algorithm}
	\SetAlgoLined
	\DontPrintSemicolon
	\label{alg}
	\begin{algorithmic}[1]
		\STATE {\bf Initialization:}
		Set the iteration number $t=1$, maximum iteration numbers $T_{max}$, threshold $\gamma_{thr}$, system sum rate $\Upsilon=0$, and initial solution of $\{\boldsymbol{P}^0,\boldsymbol{X}^0\}$.
\REPEAT [Proposed CESP algorithm]

	\STATE {\bf CESP-PA:} 
	\STATE - Solve power allocation $\boldsymbol{P}^t$ in $\eqref{PA_DC}$ with fixed previous SA outcome $\boldsymbol{X}^{t-1}$
		
	\STATE {\bf CESP-SA:}
	\STATE  - Solve subcarrier assignment $\boldsymbol{X}^t$ in $\eqref{P_SA_final}$ with fixed PA result $\boldsymbol{P}^{t}$
	
	\STATE Obtain the system sum rate $\Upsilon^{t}$ via solution of $\{\boldsymbol{X}^t,\boldsymbol{P}^t\}$
	\STATE Update the iteration counter $t\leftarrow t+1$ 

\UNTIL $| \Upsilon^{t}-\Upsilon^{t-1} | \leq \gamma_{thr}$ or $t \geq T_{max}$

\STATE Obtain the optimum solution set of subcarrier assignment and power allocation as $\{\boldsymbol{X}^*, \boldsymbol{P}^* \}= \{\boldsymbol{X}^t,\boldsymbol{P}^t\}$.
	\end{algorithmic}
\end{algorithm}

\subsection{Convergence Analysis of CESP} \label{CON_ANA}

The convergence of proposed CESP algorithm is derived as follows. Firstly, we prove that the solution obtained by D.C. approximation will converge to a local optimum. Afterwards, we conclude that the performance of proposed CESP algorithm will be improved after each iteration until converge. 

\begin{prop} \label{prop2}
	With D.C. approximation method, the respective solutions of each sub-problem for power allocation in $\eqref{PA_DC}$ and subcarrier assignment in $\eqref{P_SA_final}$ will converge to the optimum. 
\end{prop}
\begin{proof}
We first define $\mathcal{F}(\boldsymbol{P},\boldsymbol{X})$ as the sum-rate objective function which can be further expressed as the difference of two concave function $\mathcal{F}_1(\boldsymbol{P},\boldsymbol{X})-\mathcal{F}_2(\boldsymbol{P},\boldsymbol{X})$. With the employment of D.C. approximation approach, the second term $\mathcal{F}_2(\boldsymbol{P},\boldsymbol{X})$ can be asymptotically approximated to a linear function. Since the proposed CESP performs iterative schemes for PA and SA sub-problems, the optimum power allocation can be obtained with fixed subcarrier assignment, and vice versa. Therefore, we only have to prove its convergence of either sub-problem. For PA sub-problem, due to the concavity property of $\mathcal{F}_2(\boldsymbol{P})$ under fixed $\boldsymbol{X}$, we can have the following inequality at iteration $t$ as
		\begin{equation} \label{3.35}
		\mathcal{F}_2(\boldsymbol{P}^t) \leq \mathcal{F}_2(\boldsymbol{P}^{t-1})+\nabla \mathcal{F}_2(\boldsymbol{P}^{t-1}) \cdot (\boldsymbol{P}^{t}-\boldsymbol{P}^{t-1}).
		\end{equation}
		Therefore, we can obtain the following property as
		\begin{align}
		\mathcal{F}(\boldsymbol{P}^t)
		&=\mathcal{F}_1(\boldsymbol{P}^t) - \mathcal{F}_2(\boldsymbol{P}^t) \notag \\
		& \overset{\underset{(a)}{}}{\geq} \mathcal{F}_1(\boldsymbol{P}^t) - \mathcal{F}_2(\boldsymbol{P}^{t-1})-\nabla \mathcal{F}_2(\boldsymbol{P}^{t-1}) \cdot (\boldsymbol{P}^{t}-\boldsymbol{P}^{t-1}) \notag \\
		& \overset{\underset{(b)}{}}{=} \max_{\boldsymbol{P}} \ \mathcal{F}_1(\boldsymbol{P}) - \mathcal{F}_2(\boldsymbol{P}^{t-1}) - \nabla \mathcal{F}_2(\boldsymbol{P}^{t-1})\cdot(\boldsymbol{P}-\boldsymbol{P}^{t-1}) \notag \\
		& \overset{\underset{(c)}{}}{\geq} \mathcal{F}_1(\boldsymbol{P}^{t-1}) - \mathcal{F}_2(\boldsymbol{P}^{t-1})-\nabla \mathcal{F}_2(\boldsymbol{P}^{t-1})\cdot(\boldsymbol{P}^{t-1}-\boldsymbol{P}^{t-1}) \notag \\
		& = \mathcal{F}_1(\boldsymbol{P}^{t-1}) - \mathcal{F}_2(\boldsymbol{P}^{t-1}),
		\end{align}
		where (a) holds based on concavity property in $\eqref{3.35}$. The equivalent optimization problem is demonstrated in (b) at iteration $t$. On the other hand, (c) means that the sum rate with the optimal power allocation will exceed that in previous iteration $t-1$. Moreover, the optimum solution $\boldsymbol{P}^{*}$ can be obtain when $t \rightarrow \infty$, i.e., $ \lim_{t \rightarrow \infty} \left( \arg\!\max_{\boldsymbol{P}^{t}} \mathcal{F}(\boldsymbol{P}^{t})\right) = \boldsymbol{P}^{*}$. That is, the objective will either be achieved or unchanged after iteration $t$. Similarly, with D.C. approximation, the SA sub-problem also converges to the optimum. Therefore, the respective solutions of each sub-problem for power allocation in $\eqref{PA_DC}$ and subcarrier assignment in $\eqref{P_SA_final}$ will converge to the optimum. This completes the proof.
	\end{proof}

\begin{prop}
	The optimal solutions of power allocation and subcarrier assignment by using the proposed CESP algorithm demonstrated in Algotithm \ref{alg} will be improved after each iteration until convergence.
	\end{prop}
	\begin{proof}
		For CESP-PA at iteration $t+1$, we solve the power allocation $\boldsymbol{P}^{t+1}$ with fixed subcarrier assignment $\boldsymbol{X}^t$. According to Proposition \ref{prop2}, we can obtain the following inequality as
		\begin{equation} \label{prop3-1}
		\mathcal{F}(\boldsymbol{P}^t,\boldsymbol{X}^t) \leq 	\mathcal{F}(\boldsymbol{P}^{t+1},\boldsymbol{X}^t).
		\end{equation}  
Similarly, for CESP-SA, we acquire the optimal subcarrier assignment $\boldsymbol{X}^{t+1}$ with fixed power allocation result $\boldsymbol{P}^{t+1}$ as
		\begin{equation} \label{prop3-2}
		\mathcal{F}(\boldsymbol{P}^{t+1},\boldsymbol{X}^{t}) \leq 	\mathcal{F}(\boldsymbol{P}^{t+1},\boldsymbol{X}^{t+1}).
		\end{equation} 
		Hence, by combining $\eqref{prop3-1}$ and $\eqref{prop3-2}$, in can be inferred that the objective function will be achieved until convergence, and the optimum can be achieved if $t$ is sufficient large, i.e., $ \{\boldsymbol{P}^{*}, \boldsymbol{X}^{*}\} = \lim_{t \rightarrow \infty} \left( \arg\!\max_{\{\boldsymbol{P}^{t}, \boldsymbol{X}^{t}\}} \mathcal{F}(\boldsymbol{P}^{t}, \boldsymbol{X}^{t})\right)$. This completes the proof.
	\end{proof}

\section{Performance Evaluation} \label{PER_EVA}
	We evaluate the performance of proposed CESP algorithm for CoMP enhanced MNN through simulations. We consider $2$ BSs deployed with transmission radius of $100$ m serving $8$ users, where the inter-BS distance is $200$ m. Note that the users in the network include uniform and edge users, i.e., uniform users are distributed within the radius of $(0,80]$ m, whilst edge users are in that of $(80,100]$ m. We define $0\leq \eta\leq 1$ as the ratio of edge users, which means that more edge users should be served with larger values of $\eta$. Furthermore, two numerologies are configured with the number of subcarriers per numerology $N=8$ and respective ScS of $\Delta f_1=15$ kHz and $\Delta f_2=30$ kHz. The channel fading follows the pathloss model of $PL(d)=61.4+34.1\log_{10}(d)$ defined in \cite{25_5G_channel_PL}, i.e., $g_{k,m} = \sqrt{10^{\left(-PL(d)/10\right)}}$; while $\beta_{k,m,n}^{i} \sim \exp(1)$ indicates the Rayleigh fading effect following exponential distribution. Note that users are considered to have the same requirements, i.e., $P^{max}_{k}=P^{max}$, $r^{req}_{m}=r^{req}$, and $l^{req}_{m}=l^{req}$, $\forall k\!\in\!\mathcal{K}, m\!\in\!\mathcal{M}$. The remaining system parameters are listed in Table \ref{syspara}. We first evaluate the performance of proposed CESP algorithm in terms of its convergence and system sum rate under different maximum allowable transmit power, QoS requirements and latency restrictions. Then, we compare the performance of CoMP enhanced multi-numerology with that of non-CoMP and single-numerology mechanisms with different users distributions including uniform and edge users. Furthermore, we compare our proposed CESP scheme with other benchmarks in existing literatures.

\begin{table}[!ht]
\small
	\centering
	\caption {System Parameters of CoMP enhanced MNN}
	\begin{adjustbox}{max width=1\textwidth}
		\begin{tabular}{ll}
			\hline
		System Parameters & Value\\ \hline \hline
		BS serving coverage radius & 100 m\\
		Operating center frequency & 3.5 GHz \\
		Number of BSs & 2\\
		Number of users & 8\\
		Number of numerology & $2$\\
		Number of subcarriers per numerology & 8\\			
		ScS of numerology 1 and 2 & $\{15,30\}$ kHz\\
		Noise power & $-90$ dBm\\
		Maximum transmit power of BS & $23$ dBm\\			
		Latency of numerology 1 and 2 & $\{0.5,1\}$ ms\\
		Latency requirement & $0.75$ ms\\		
		Penalty factor & $10^3$\\
		Threshold of convergence & $0.1$ b/s/Hz \\
		Maximum allowable iterations & $100$ \\
			\hline
		\end{tabular} \label{syspara}
	\end{adjustbox}
\end{table}

\subsection{Convergence and Sum-Rate Performance of CESP}
As shown in Fig. \ref{Fig.1}, we have exhibited the convergence of proposed CESP algorithm considering different maximum transmit power $P^{max}$. We can observe from the figure that the sum rate quickly converges and becomes saturated with around $6$ iterations.  Additionally, it can be inferred that we have higher system sum rate using higher transmit power, i.e., it achieves sum rate of around $132$ b/s/Hz under $P^{max}=30$ dBm. This is because that the proposed CESP algorithm possesses higher degree of freedom of candidate power allocation to mitigate strong interferences of INI and ICI, which leads to higher SINR values.

\begin{figure}
	\centering
	\includegraphics[width=3.1in]{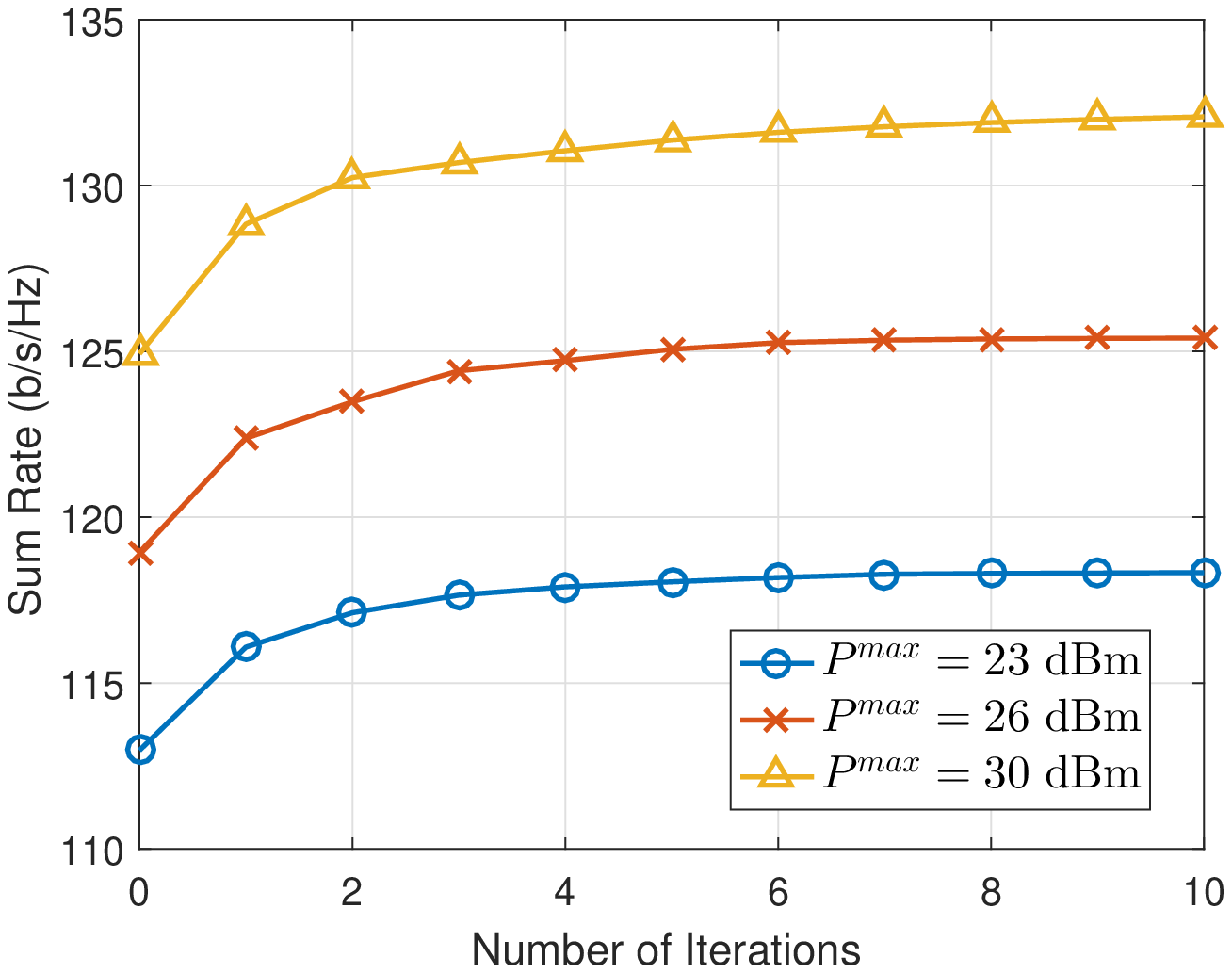}
	\caption{Convergence of proposed CESP algorithm with different maximum transmit power $P^{max}=\{23, 26, 30\}$ dBm, ratio of edge users $\eta=0.5$, QoS requirement $r^{req}=1$ b/s/Hz and latency constraint $l^{req}=0.75$ ms.}
	\label{Fig.1}
\end{figure}

Fig. \ref{Fig.2} depicts the sum rate performances versus different maximum transmit power considering different QoS requirements $r^{req}$ and latency constraints $l^{req}$ which are demonstrated in Figs. \ref{fig_2_1} and \ref{fig_2_2}, respectively. As explained previously, larger transmit power is beneficial to provide higher system sum rate. As shown in Fig. \ref{fig_2_1}, we notice that higher QoS demands $r^{req}$ will lead to lower sum rate performance due to stringent service requirements. In CoMP enhanced MN network, the BSs will potentially require more resources to serve those highly-demanding edge users. However, if no QoS required, i.e., $r^{req}=0$ b/s/Hz, the system will tends to allocate resources to users capable of achieving the highest rate due to the original sum-rate maximization problem. As demonstrated in Fig. \ref{fig_2_2}, it reveals that system sum rate degrades with decrements of $l^{req}$, i.e., stringent latency will confine the system performance for CoMP enhance MNN. The BSs may require more resource blocks of MN with shorter timeslots to satisfy the latency requirement. Furthermore, in order to investigate the effects of maximum transmit power and QoS constraints, we define the data rate outage as
\begin{equation} \label{Dout}
\Upsilon_{out} = 1 - \frac{1}{M} \sum_{m\in \mathcal{M}} \mathbbm{1} \left( \sum_{\forall i \in \mathcal{I}} \sum_{\forall k \in \mathcal{K}} \sum_{\forall n \in \mathcal{N}} r_{k,m,n}^{i} \geq r^{req} \right).
\end{equation}
As depicted in Fig. \ref{Fig.3}, we evaluate data rate outage under different maximum transmit power and QoS requirements. We can observe that it has higher outage of about $\Upsilon_{out} = 0.38$ under $r^{req}=7$ b/s/Hz and $P^{max}=23$ dBm. This is because resources are insufficient to support highly-demanding users. However, it has a decreasing trends of data rate outage with escalating transmit power from $P^{max}=23$ to $36$ dBm. It also demonstrates that zero outage can be achieved where both normal and edge users meet their demands, if more transmit power is assigned under rigorous QoS requirement of $7$ b/s/Hz.

\begin{figure}[htbp]
	\centering
	\subfigure[]{
		\includegraphics[width=3.1in]{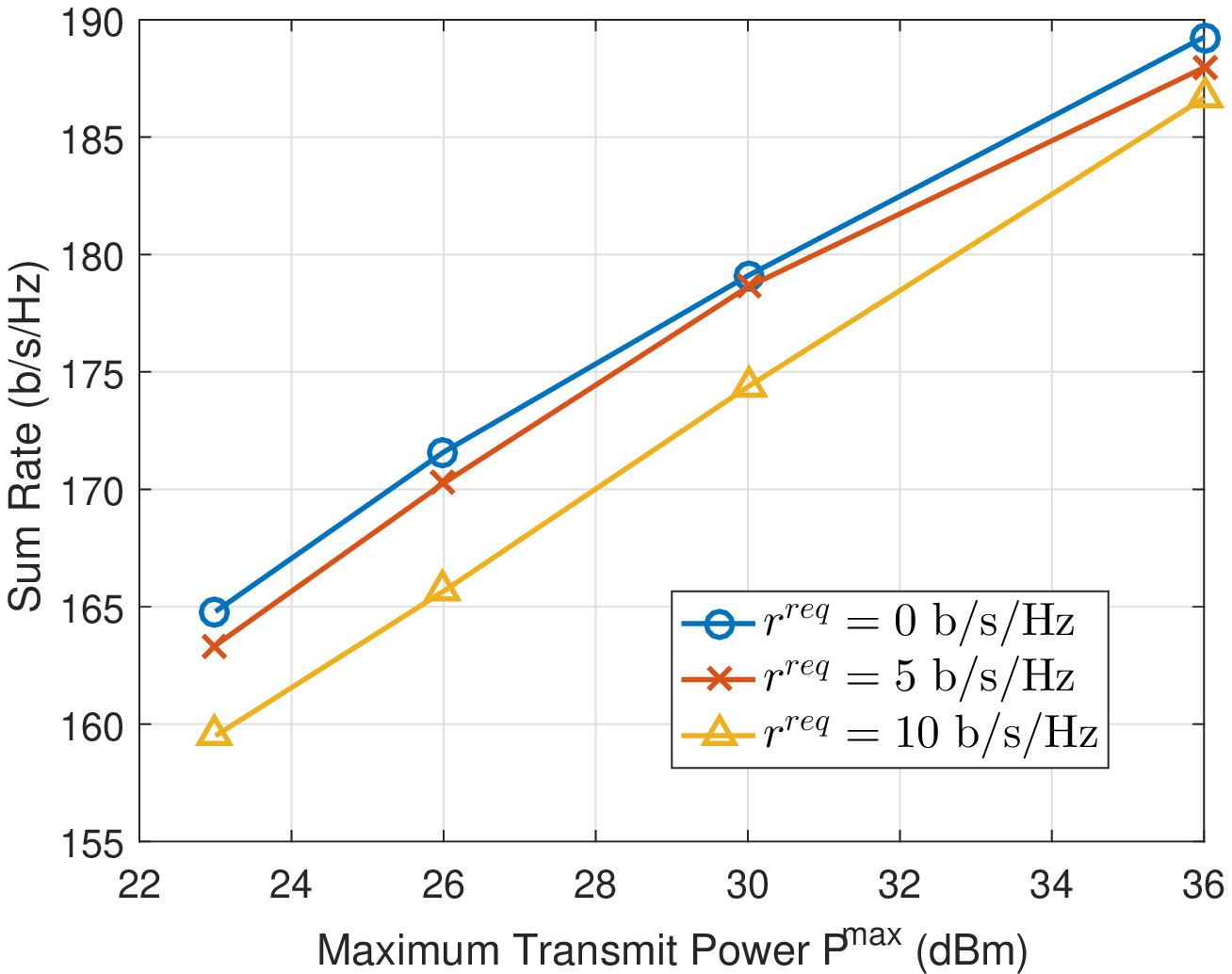}
		\label{fig_2_1}	}
	\subfigure[]{
		\includegraphics[width=3.1in]{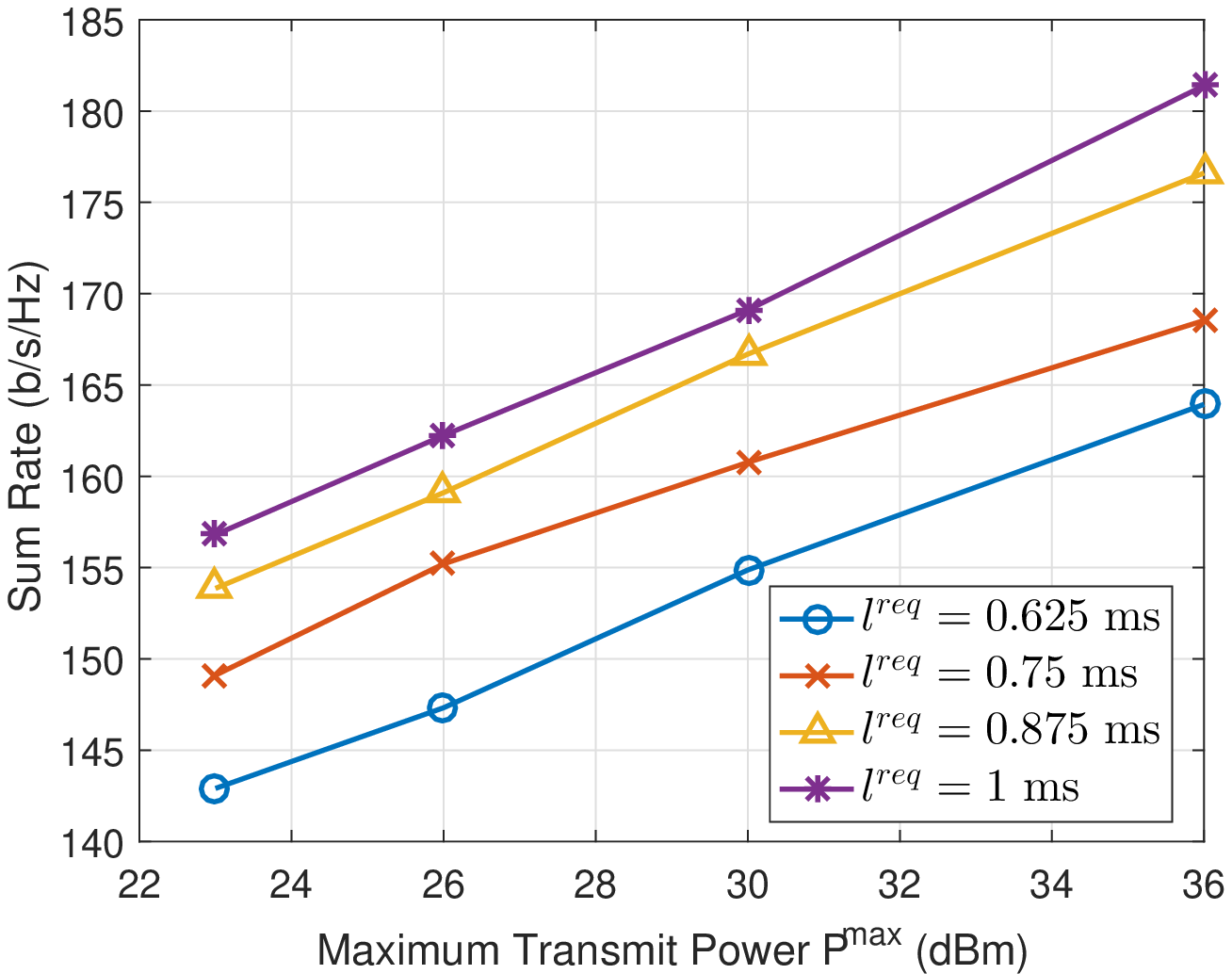}
		\label{fig_2_2}}
	\caption{Performance of sum rate of CESP
	 versus maximum transmit power $P^{max}$ under various (a) QoS requirements $r^{req}=\{0,5,10\}$ b/s/Hz and (b) latency constraints $l^{req} = \{0.625,0.75,0.875,1\}$ ms with radio of edge user $\eta=0.5$.}\label{Fig.2}
\end{figure}

\begin{figure}[htbp]
	\centering
	\includegraphics[width=3.3in]{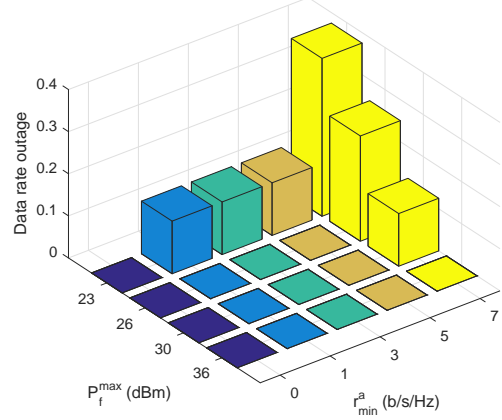}
	\caption{Data rate outage of CESP under different maximum transmit power $P^{max}=\{23,26,30,36\}$ dBm and QoS requirement $r^{req}=\{0,1,3,5,7\}$ b/s/Hz with ratio of edge users $\eta=0.5$ and latency constraint $l^{req}=0.75$ ms.}
	\label{Fig.3}
\end{figure}

\subsection{Effect of Different User Distributions}

\begin{figure}[htbp]
	\centering
	\subfigure[]{
		\includegraphics[width=3.1in]{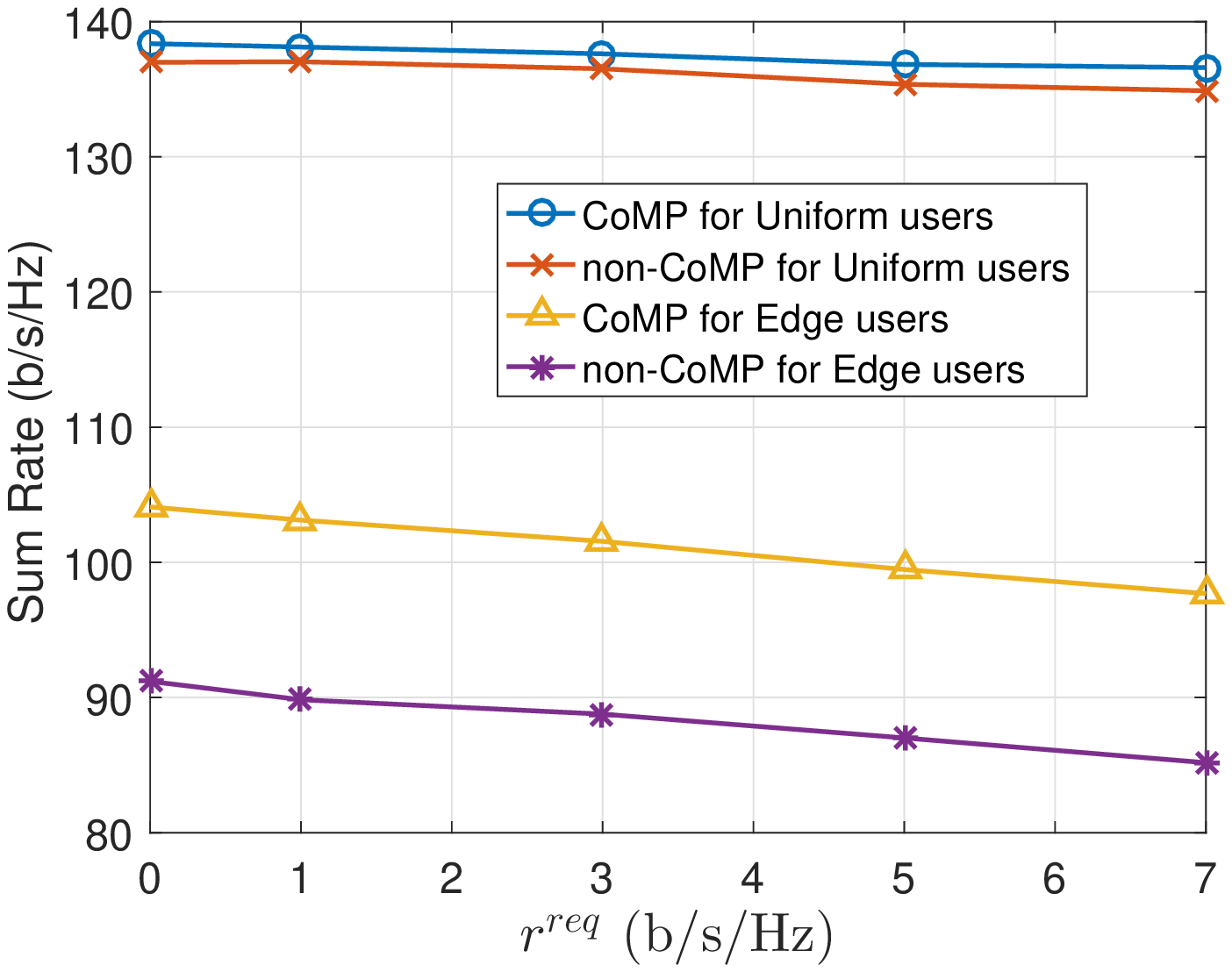}
		\label{Fig.4}	}
	\subfigure[]{
		\includegraphics[width=3.1in]{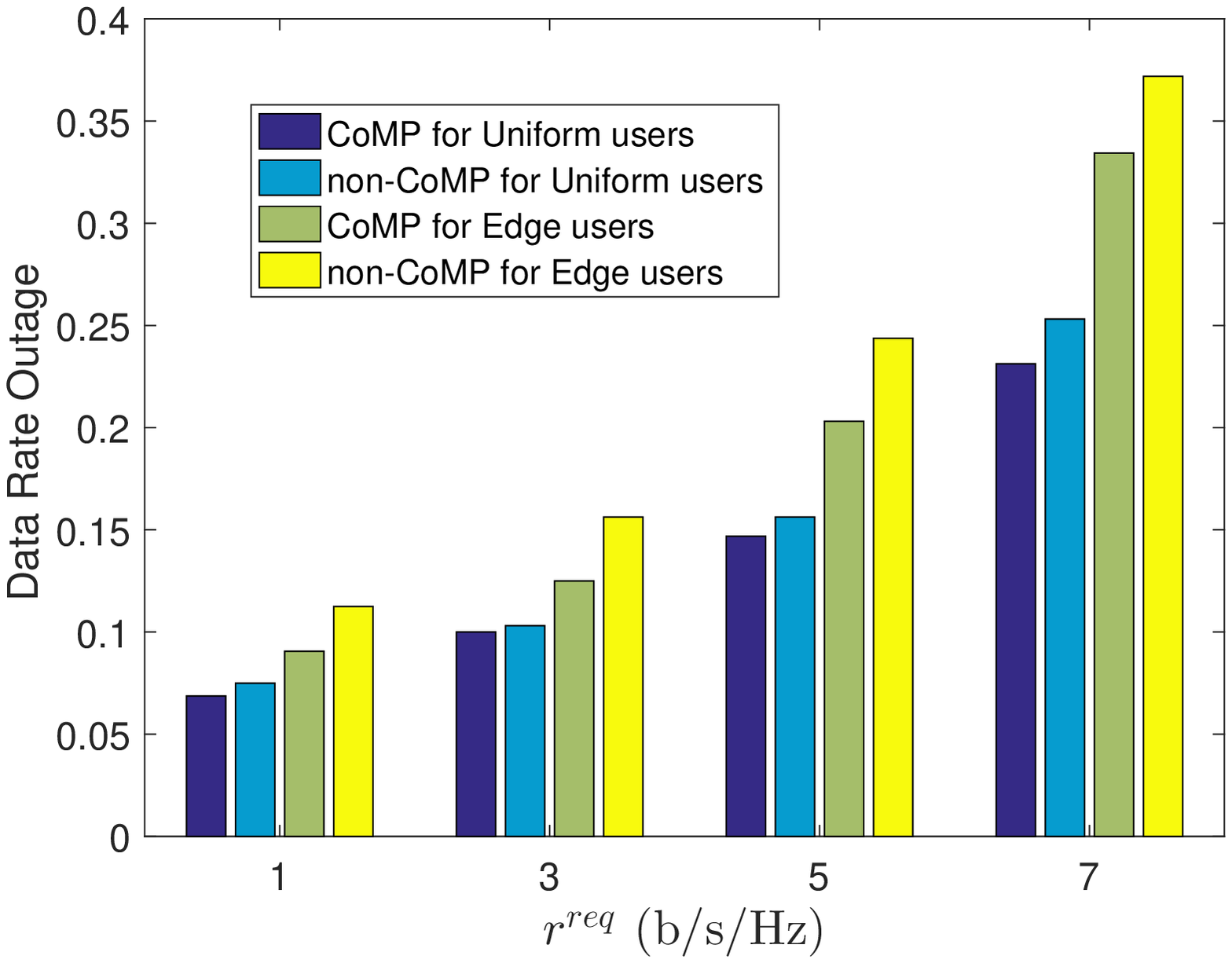}
		\label{Fig.5}}
	\caption{Performance of proposed CESP algorithm in terms of (a) sum rate and (b) data rate outage versus different QoS requirement $r^{req}$ with maximum transmit power $P^{max}=23$ dBm, ratio of edge users $\eta=0.5$ and latency constraint $l^{req}=0.75$ ms.}
	\label{Fig.45}
\end{figure}

We consider two types of user distributions including normal and edge users, which are defined at the beginning of this section. The edge users have comparably lower throughput performance due to ICI from non-serving BSs. However, CoMP transmission possesses the capability to improve SINR by transforming interfered BSs into coordinated ones. Therefore, we evaluate CoMP for uniform and edge users compared with conventional non-CoMP techniques, which are demonstrated in Fig. \ref{Fig.45}. Note that under conventional non-CoMP technique the user can be only served by a single BS. As shown in Fig. \ref{Fig.4}, we can observe that higher QoS requirements lead to lower sum rate performance due to rigorous service demands. For edge users, the proposed CESP algorithm under CoMP transmissions achieves higher sum rate compared to non-CoMP mechanism with a difference of around $15$ b/s/Hz. This is because strong ICI is potentially alleviated through coordinated joint transmissions among BSs. Moreover, we can observe that adoption of CoMP will provide higher performance increase on edge users, i.e., around $14$ b/s/Hz more on sum rate of CoMP than that with non-CoMP under different QoS requirements. As depicted in Fig. \ref{Fig.5}, we further evaluate data rate outage performance in $\eqref{Dout}$ with different QoS demands comparing CoMP/non-CoMP for uniform/edge users. It can be inferred that higher QoS requirements potentially provoke higher outage due to insufficient subcarrier and power resources that can be assigned to edge users. Moreover, benefited by less ICI, the lowest data rate outage is achieved for uniform users with CoMP mechanism under different QoS requirements. On the other hand, with the adoption of CESP scheme, the ICI interferences can be effectively mitigated by the CoMP mechanism especially for edge users. Much lower data rate outage is obtained for CoMP-enabled edge users compared to non-CoMP transmissions, e.g., around $5\%$ less outage under $r^{req} = 7$ b/p/Hz.

\subsection{Comparison of CoMP/non-CoMP with MN/SN}

\begin{figure}[!t]
	\centering
	\subfigure[]{
		\includegraphics[width=3in]{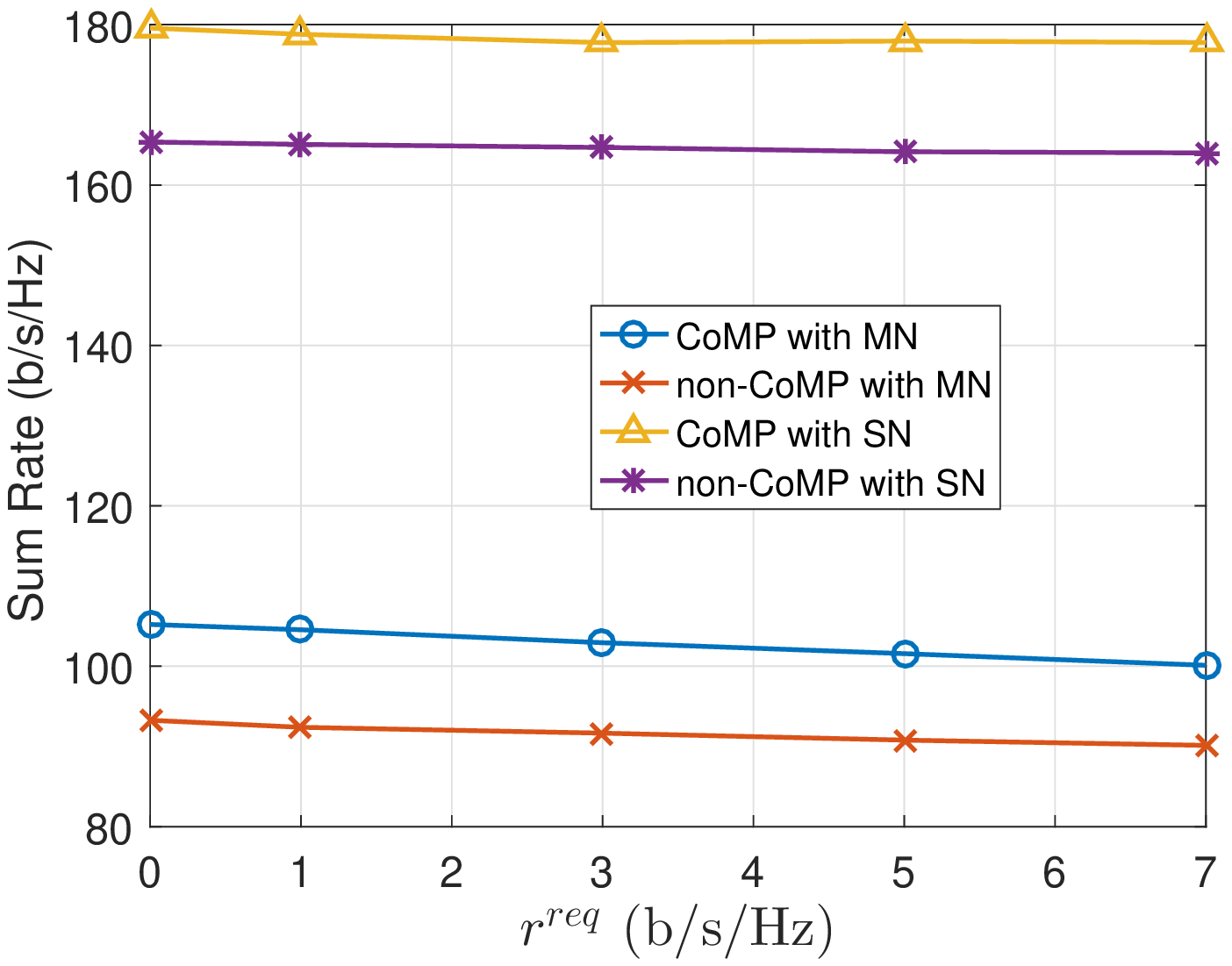}
		\label{fig_6_1}}
	\subfigure[]{
		\includegraphics[width=3in]{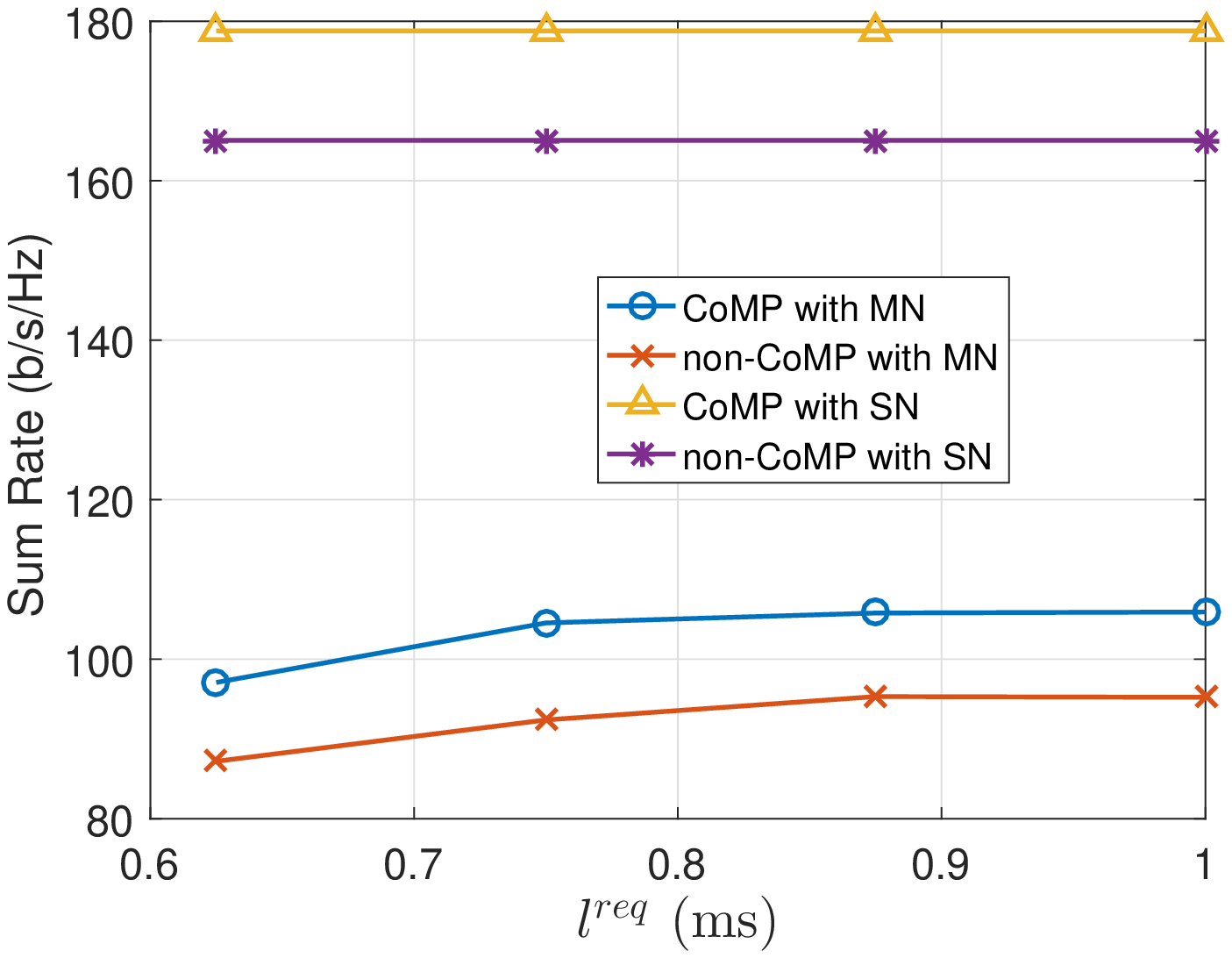}
		\label{fig_6_2}}
	\caption{The sum rate performance of proposed CESP algorithm with different (a) QoS requirements $r^{req}$ and (b) latency constraints $l^{req}$ under scenarios of CoMP/non-CoMP for MN/SN considering maximum transmit power $P^{max}=23$ dBm and ratio of edge users $\eta=0.5$.}\label{Fig.6}
\end{figure}

\begin{figure}[!t]
	\centering
	\includegraphics[width=3.3in]{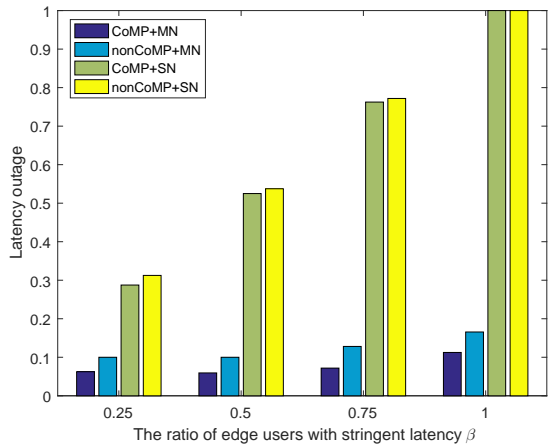}
	\caption{Latency outage of CESP considering different ratios of edge users $\eta=\{0.25,0.5,0.75,1\}$ with maximum transmit power $P^{max}=23$ dBm, QoS requirement $r^{req}=5$ b/s/Hz and latency constraint $l^{req}=0.75$ ms.}
	\label{Fig.7}
\end{figure}

	As shown in Fig. \ref{Fig.6}, we evaluate the system performance of proposed CESP algorithm considering MN system comparing with single-numerology (SN) structure under CoMP and non-CoMP transmissions with respect to different QoS and latency constraints. Note that we consider conventional structure for SN, i.e., the ScS is equal to $15$ kHz if configured for all resource blocks. We can observe from Fig. \ref{fig_6_1} that MN has lower sum rate than SN due to the existence of INI, which leads to a degradation of about $70$ b/s/Hz. Note that INI is induced by non-orthogonality property between different numerologies. However, it will be beneficial to employ joint transmission of CoMP to mitigate interference, which improves sum rate performance of about $15$ b/s/Hz in both MN and SN networks. Furthermore, we can observe from Fig. \ref{fig_6_2} that sum rate of MN networks escalates with looser latency restrictions, whilst unchanged performance is shown in SN system. This is because MN is capable of flexibly assigning resources with different numerologies to fulfill various latency requirements, whilst SN cannot provide diverse service demands due to its fixed structure. Therefore, we study the significance of MN in terms of latency outage performance which is defined as
\begin{equation} \label{Lout}
l_{out} = 1 - \frac{1}{M} \sum_{m\in \mathcal{M}} \mathbbm{1} \left( l_{m} \leq l^{req}_{m} \right),
\end{equation}
where the indicator function represents whether latency constraint is satisfied by the users. In Fig. \ref{Fig.7}, the latency outage with different ratios of edge users of $\eta$ is performed. We can intuitively know that higher latency outage occurs due to insufficient resources to serve increased number of edge users. Furthermore, we can infer from the figure that SN possesses the highest outage due to its incapability to support different time durations, e.g., a full latency outage takes places if SN is configured to serve all edge users. However, under MN enhanced system, flexible resource blocks with various lengths of timeslots can substantially reduce the latency outage of the CoMP case to become only around $0.15$ for $\eta=1$. 


\begin{figure}[!t]
	\centering
	\subfigure[$r^{req} = 1$ b/s/Hz]{
		\includegraphics[width=3in]{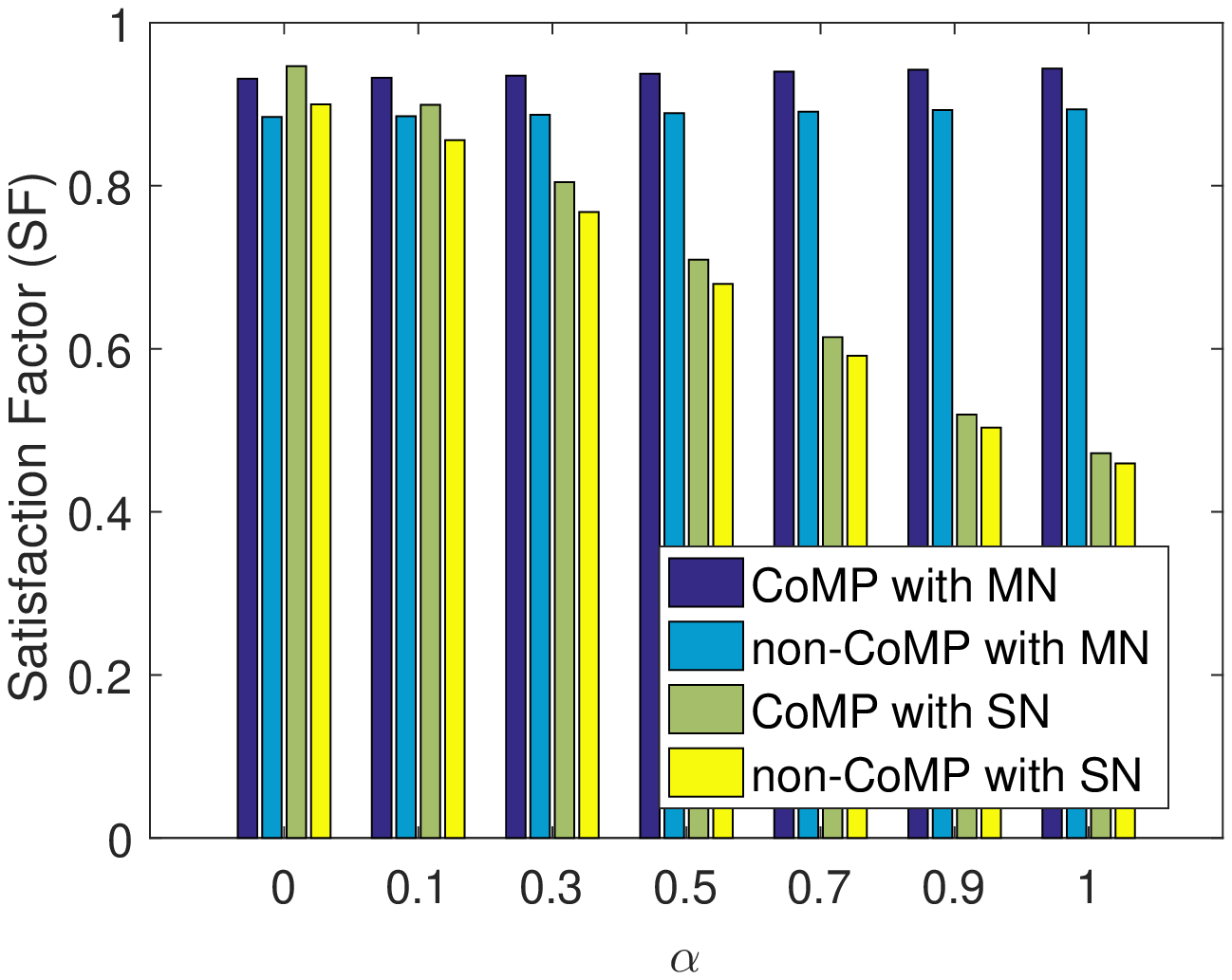}
		\label{sub1}}
	\subfigure[$r^{req} = 3$ b/s/Hz]{
		\includegraphics[width=3in]{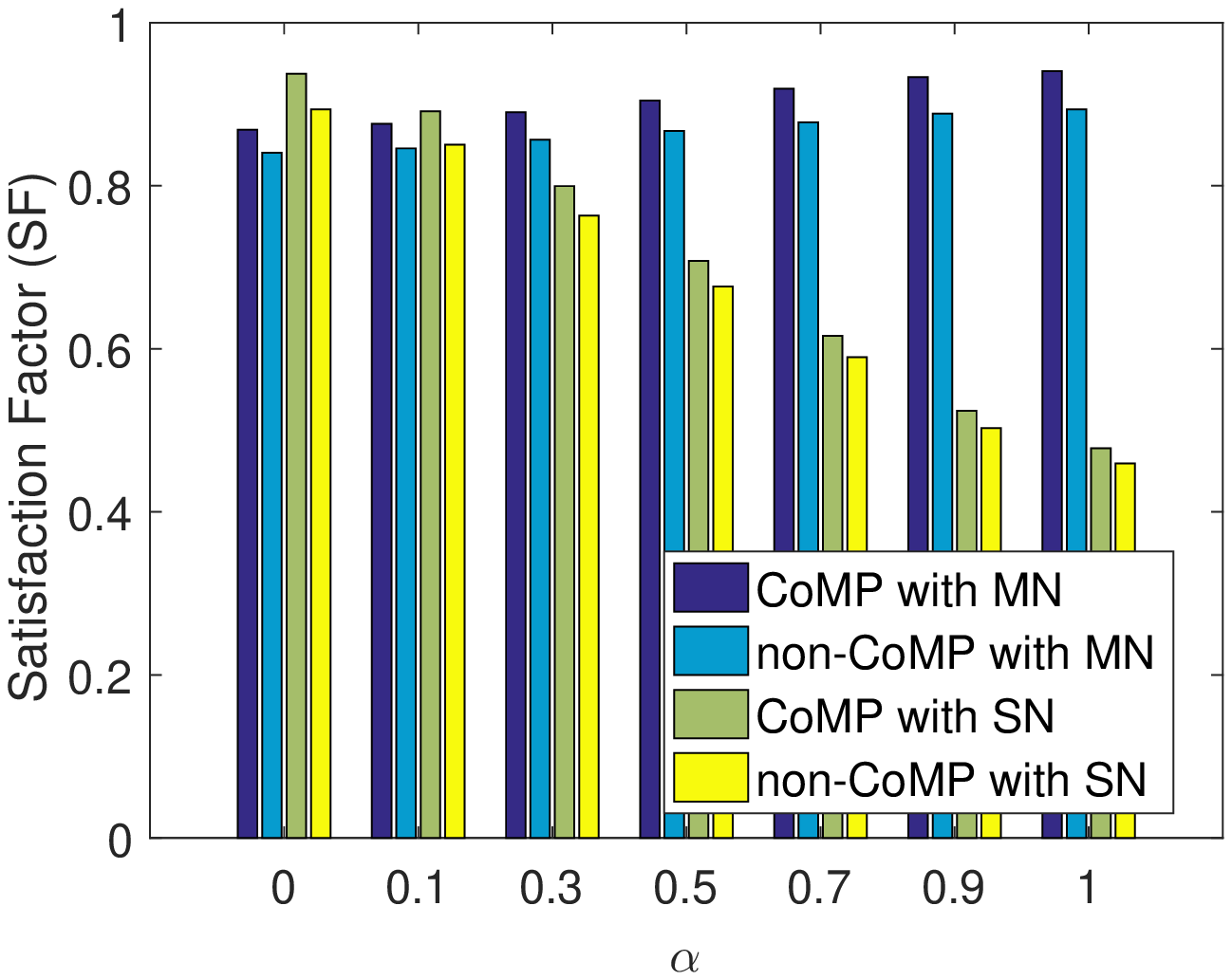}
		\label{sub2}}
	\subfigure[$r^{req} = 5$ b/s/Hz]{
		\includegraphics[width=3in]{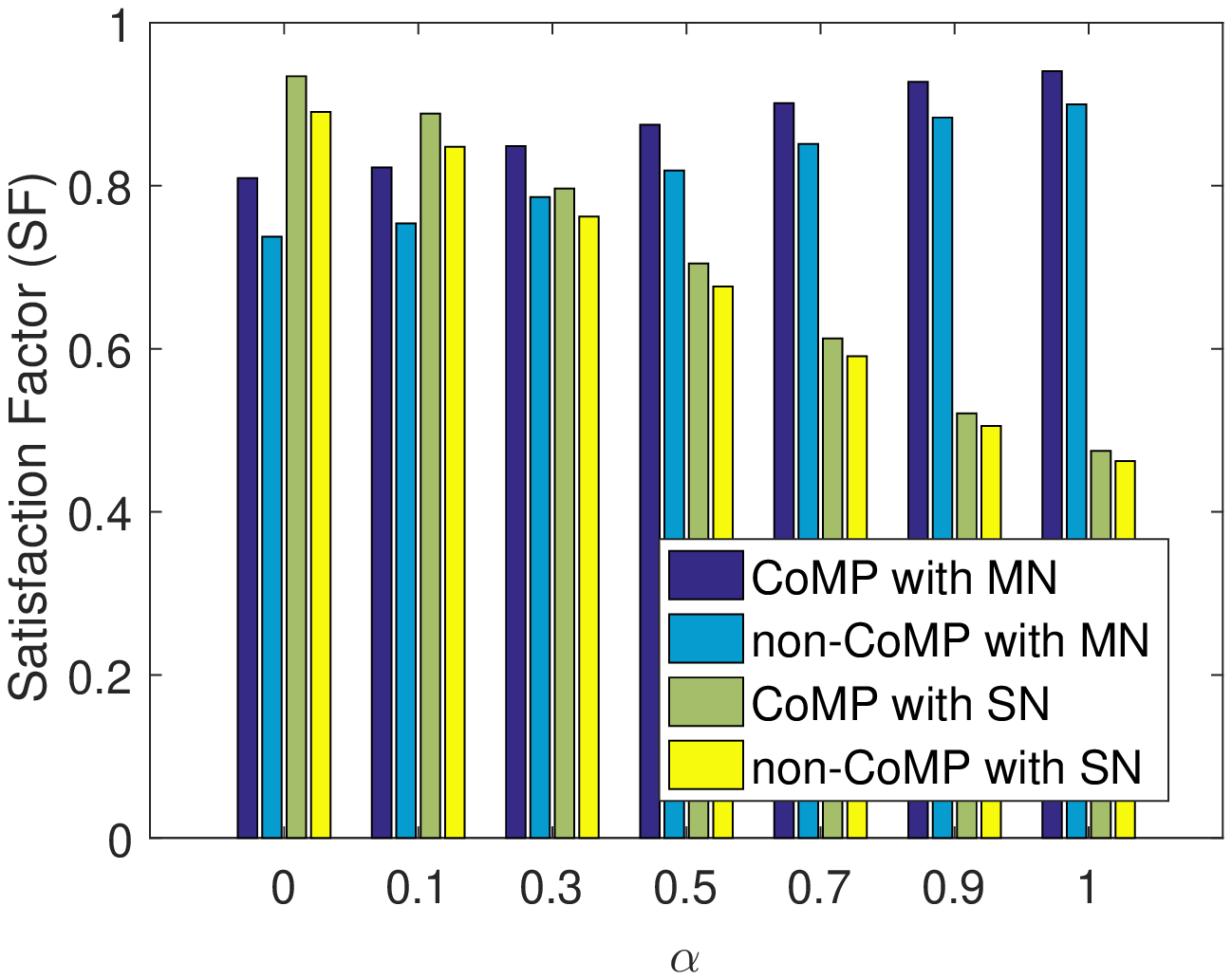}
		\label{sub3}}
	\subfigure[$r^{req} = 7$ b/s/Hz]{
		\includegraphics[width=3in]{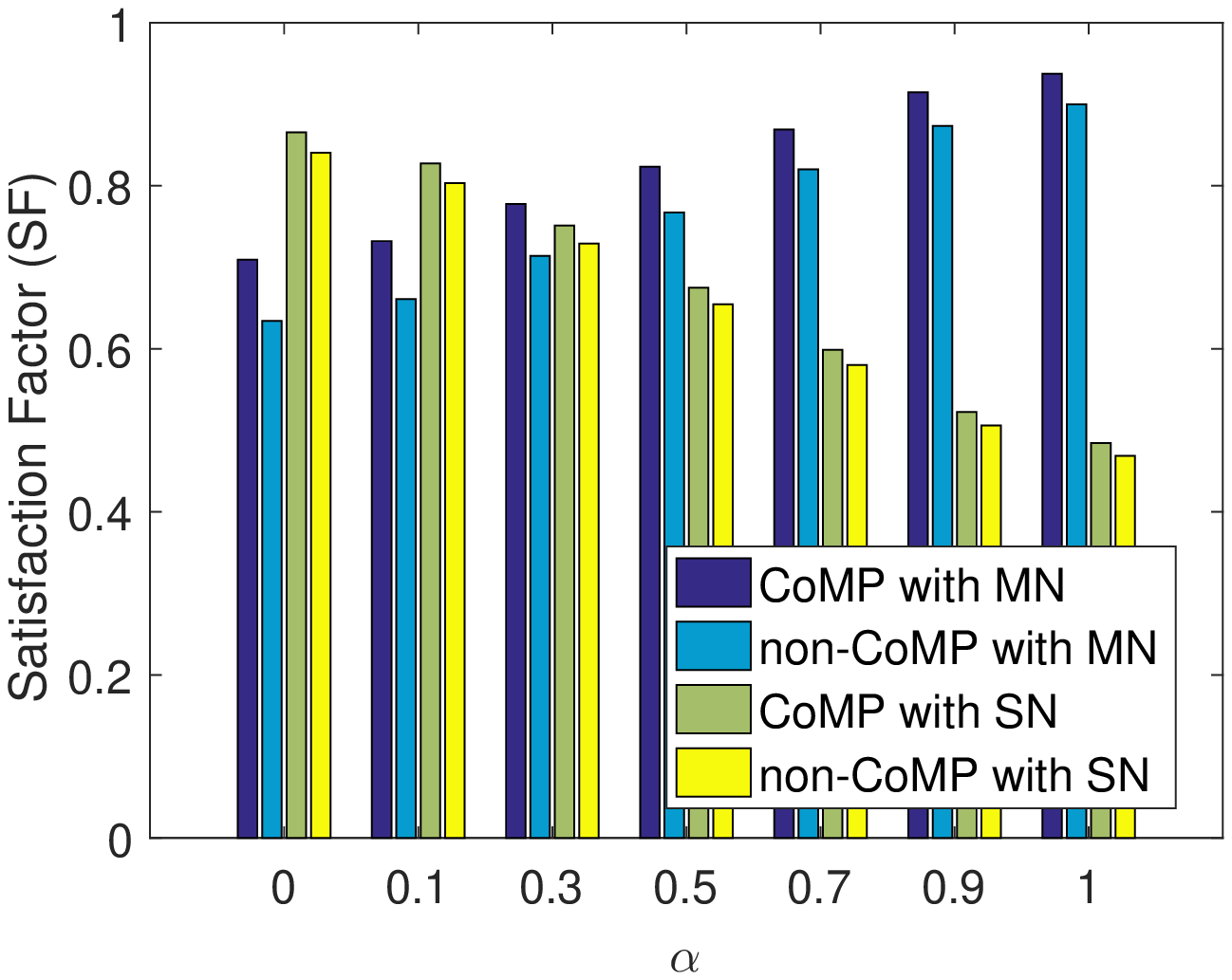}
		\label{sub4}}
	\caption{Performance of satisfaction factor of proposed CESP algorithm considering various weights of $\alpha=\{0,0.1,0.3,0.5,0.7,0.9,1\}$ under different QoS requirements $r^{req}=\{ 1,3,5,7\}$ b/s/Hz with maximum transmit power $P^{max}=23$ dBm, latency constraint $l^{req}=0.75$ ms and ratio of edge users $\eta=0.5$.}\label{Fig.8}
\end{figure}

In order to investigate the joint influence of proposed CESP algorithm in terms of QoS requirements in $\eqref{Dout}$ and latency outage in $\eqref{Lout}$, we further define the satisfaction factor (SF) as
\begin{equation}\label{SF}
SF = (1-\alpha)\cdot \Upsilon_{out} + \alpha\cdot l_{out},
\end{equation}
where $\alpha$ is a constant implying the emphasis of a system between data rate requirement and latency constraint. We can observe from $\eqref{SF}$ that $SF$ of CESP is dominated by latency constraint with increasing $\alpha$ establishing a latency-aware network, whilst it provides QoS-aware services with decrement of $\alpha$. In Fig. \ref{Fig.8}, we evaluate the satisfaction $SF$ of proposed CESP algorithm versus different weights of $\alpha$ under various QoS requirements. Firstly, we can observe from Fig. \ref{sub1} to \ref{sub4} that with increasing QoS requirements under smaller $\alpha=\{0,0.1,0.3\}$, all of the four schemes result in decreasing $SF$ due to insufficient resource assignment. Furthermore, with larger values of $\alpha=\{0.5,0.7,0.9,1\}$, MN achieves comparably higher $SF$ than that of SN since the network has a tendency to offer latency-aware services with low latency outage. Under low QoS requirements, i.e., $r^{req}=1$ b/s/Hz as depicted in Fig. \ref{sub1}, the proposed CESP for MN networks sustains stable $SF$ performance striking a compelling balance between QoS and latency owing to its flexible allocation among different numerologies. On the other hand, with $r^{req}=7$ b/s/Hz as shown in Fig. \ref{sub4}, the performance of CESP under CoMP enhanced MN increases from $SF=0.72$ to $0.93$, whilst that of SN degrades from $SF=0.85$ to $SF=0.49$ which has much lower satisfaction values compared to MN. Additionally, benefited by CoMP enhanced transmissions, higher $SF$ can be achieved in both MN and SN networks due to the mitigation of ICI. Although CoMP enhanced SN is more appropriate than MN due to the non-existence of INI for QoS-oriented network under $\alpha=0$, SN lacks the flexibility to adjust numerology lengths of resource blocks which restricts the capability to fulfill different network service requirements.

\subsection{Comparison of Proposed CESP Algorithm}

\begin{figure}[!t]
	\centering
	\includegraphics[width=3.3in]{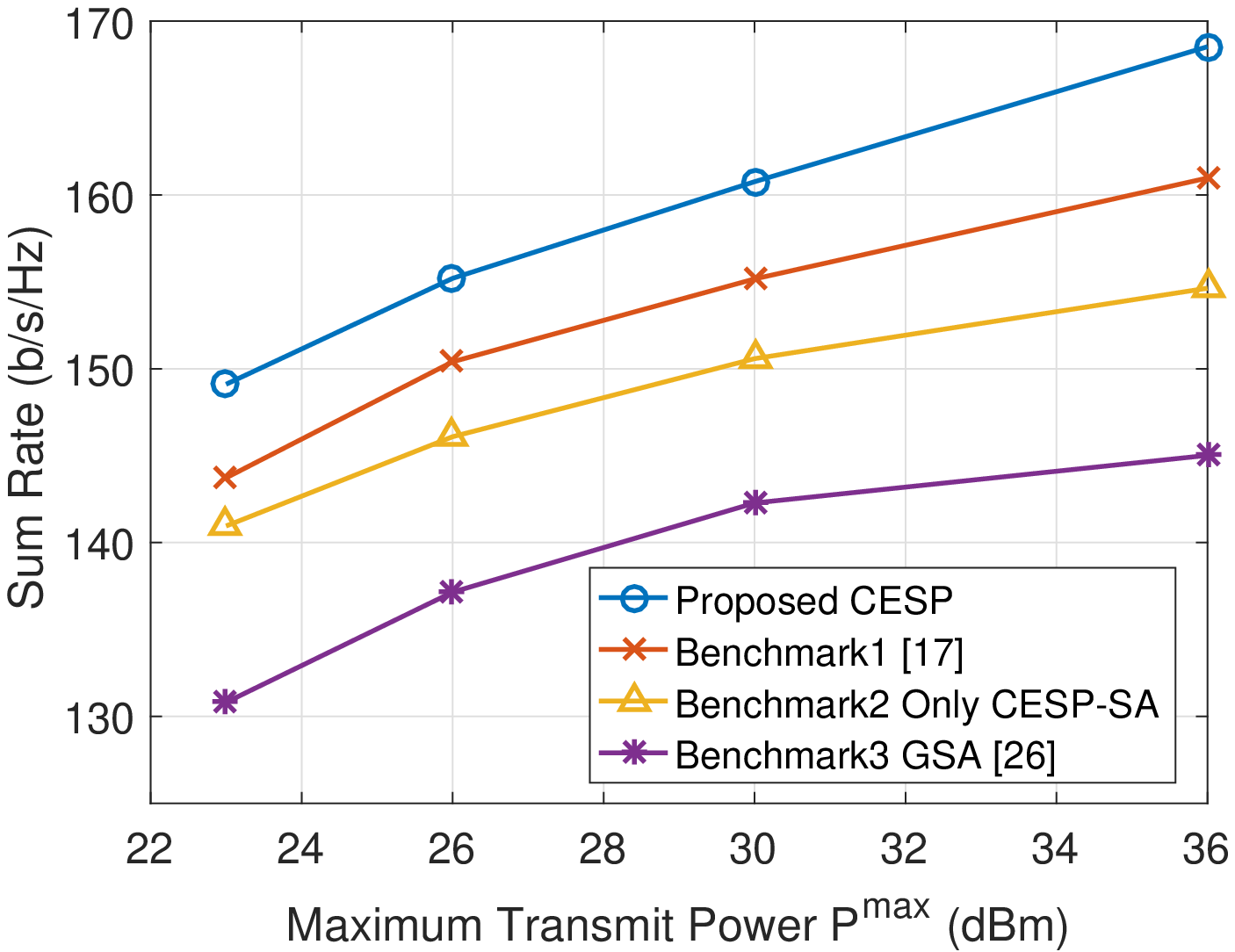}
	\caption{Comparison of proposed CESP algorithm and existing benchmarks for CoMP enhanced MNN in terms of sum rate considering maximum transmit power $P^{max}=23$ dBm, QoS requirement $r^{req}=1$ b/s/Hz, latency constraint $l^{req}=0.75$ ms and ratio of edge users $\eta=0.5$.} 
	\label{Fig.9}
\end{figure}

As demonstrated Fig. \ref{Fig.9}, we compare our proposed CESP algorithm with the existing benchmarks in open literatures as follows: \textit{Benchmark1} referring \cite{18_main_ref} adopts equal power allocation (EPA); whereas the subcarrier assignment sub-problem is resolved by integer relaxation, i.e., the discrete variable of $x^{i}_{k,m,n}$ is relaxed to a continuous range of $[0,1]$. However, constraint $\eqref{SA new_C}$ is not adopted in \cite{18_main_ref} which potentially leads to quantization errors when transforming continuous variables into discrete indexes of subcarrier assignment. \textit{Benchmark2} applies EPA along with our CESP-SA scheme in Subsection \ref{SA_sec} for subcarrier allocation problem. On the other hand, \textit{Benchmark3} employs CESP-PA in Subsection \ref{PA_sec} for power allocation; while genetic subcarrier assignment (GSA) in \cite{100} is utilized which adopts discrete genetic algorithm to allocates the optimal subcarrier indicators. We can observe from Fig. \ref{Fig.9} that higher sum rate can be achieved with higher allowable maximum transmit power which potentially provides higher degree of freedom to assign resources to mitigate both INI and ICI, i.e., the sum rate of CESP increases from $149$ to $169$ b/s/Hz from $P^{max}=23$ to $26$ dBm. Furthermore, it can be seen that the \textit{Benchmark3} with GSA displays the worst performance since the local optimum may be obtained by using the genetic algorithm. Since \textit{Benchmark2} adopts the proposed CESP-SA algorithm, it can acquire the optimal subcarrier allocation achieving higher sum rate performance than GSA. Furthermore, as for \textit{Benchmark1}, it has a degradation of around $8$ b/s/Hz compared to that of proposed CESP algorithm due to its quantization error induced from the process of integer relaxation. To summarize, under the CoMP enhanced MNN, our proposed CESP algorithm can optimally and flexibly assign subcarrier and power resources outperforming the other existing benchmarks in terms of sum rate while sustaining latency constraints.

\section{Conclusions} \label{SEC_CON}

We have conceived a CoMP enhanced MNN for 5G-NR wireless network. We have proposed a CESP scheme which aims at flexibly assigning power and subcarrier resources to maximize system sum rate constrained by QoS and latency requirement. A solvable convex optimization problem is theoretically proved and obtained with the employment of D.C. approximation, integer relaxing and Lagrangian properties. By mitigating ICI and INI effects, the simulation results have demonstrated that the proposed CESP scheme for CoMP enhanced MN transmission can perform higher sum rate and lower rate and latency outages for both uniform and edge users. It strikes a compelling balance between QoS and latency outage offering either QoS-oriented or latency-aware services. Moreover, our proposed CESP scheme can optimally and flexibly assign subcarrier and power resources, which outperforms the existing benchmarks in open literatures in terms of sum rate while sustaining latency constraints.

\bibliographystyle{IEEEtran}
\bibliography{IEEEabrv}
\end{document}